\documentclass[
a4paper
]{article}

\usepackage[utf8]{inputenc}
\usepackage{xspace}
\usepackage{graphicx}
\usepackage{algorithm}
\usepackage[noend]{algpseudocode}

\usepackage{amsmath}
\usepackage{amssymb}
\usepackage{amsthm}

\usepackage{MnSymbol}
\usepackage{mathtools}
\usepackage{textcomp}
\usepackage{fullpage}

\theoremstyle{plain}
\newtheorem{theorem}{Theorem}[section]
\newtheorem*{theorem*}{Theorem}
\newtheorem{proposition}[theorem]{Proposition}
\newtheorem*{proposition*}{Proposition}
\newtheorem{lemma}[theorem]{Lemma}
\newtheorem*{lemma*}{Lemma}

\newtheorem*{corollary*}{Corollary}

\newtheorem*{claim*}{Claim}

\theoremstyle{definition}

\newtheorem*{definition*}{Definition}

\theoremstyle{remark}

\newtheorem*{observation*}{Observation}
\newtheorem{example}[theorem]{Example}
\newtheorem*{example*}{Example}
\usepackage{tikz}
\usepackage{varwidth}

\DeclareMathAlphabet{\pazocal}{OMS}{zplm}{m}{n}

\usetikzlibrary{automata,angles,calc,graphs,positioning,chains,arrows,decorations.pathmorphing}

\algrenewcommand\algorithmicindent{1.3em}%
\algnewcommand\algorithmicforeach{\textbf{for each}}
\algdef{S}[FOR]{ForEach}[1]{\algorithmicforeach\ #1\ \algorithmicdo}

\renewcommand{\paragraph}[1]{\smallskip \noindent \textbf{#1.}}

\newcommand{\changed}[1]{#1}

\newcommand{\nat}{\mathbb{N}}

\newcommand{\cL}{\pazocal{L}}
\newcommand{\bbN}{\mathbb{N}}
\newcommand{\cO}{\pazocal{O}}

\newcommand{\sem}[1]{{\lsem{}{#1}\rsem}}
\newcommand{\semaux}[1]{\lceil\hspace{-0.8mm}\lceil{#1}\rfloor\hspace{-0.8mm}\rfloor}
\newcommand{\dom}{\operatorname{dom}}
\newcommand{\Stream}{\pazocal{S}}

\newcommand{\rt}[1]{\texttt{root}(#1)}
\newcommand{\depth}[2]{\texttt{depth}_{#1}(#2)}

\newcommand{\parent}[2]{\texttt{parent}_{#1}(#2)}
\newcommand{\children}[2]{\texttt{children}_{#1}(#2)}
\newcommand{\desc}[2]{\texttt{desc}_{#1}(#2)}
\newcommand{\ancst}[2]{\texttt{ancst}_{#1}(#2)}
\newcommand{\leaves}[2]{\texttt{leaves}_{#1}(#2)}

\newcommand{\Hom}{\operatorname{Hom}}

\newcommand{\atoms}[1]{\textit{atoms}(#1)}
\newcommand{\exatom}[1]{R_{#1}(\bar{x}_{#1})}
\newcommand{\mult}[2]{\operatorname{mult}_{#1}(#2)}

\newcommand{\T}{\textbf{T}}
\newcommand{\Data}{\textbf{D}}
\newcommand{\Var}{\textbf{X}}

\newcommand{\arity}{\operatorname{arity}}
\newcommand{\tuples}{\operatorname{Tuples}}
\newcommand{\Schema}{\sigma}
\newcommand{\bleft}{\{\!\!\{}
\newcommand{\bright}{\}\!\!\}}

\newcommand{\prefix}{\preceq_p}

\newcommand{\Sigmastar}{\Sigma^\ast}
\newcommand{\chain}{CCEA\xspace}

\newcommand{\cA}{\pazocal{A}}

\newcommand{\cC}{\pazocal{C}}

\newcommand{\anamelc}{parallelized finite automata\xspace}
\newcommand{\aacro}{PFA\xspace}

\newcommand{\anamecea}{Parallelized-CEA\xspace}
\newcommand{\anamelong}{Parallelized Complex Event Automata\xspace}
\newcommand{\acrocea}{PCEA\xspace}

\newcommand{\cP}{\mathcal{P}}

\newcommand{\un}{\textbf{U}}

\newcommand{\CHCEA}{\pazocal{C}}
\newcommand{\ACH}{\pazocal{P}}
\newcommand{\bin}{\textbf{B}}
\newcommand{\uncq}{\textbf{U}_{\operatorname{lin}}}
\newcommand{\bincq}{\textbf{B}_{\operatorname{eq}}}
\newcommand{\leftbinfunc}[1]{\cev{#1}}
\newcommand{\rightbinfunc}[1]{\vec{#1}}

\newcommand{\binfunc}{\mathcal{B}}

\newcommand{\outtime}{\mathsf{time}}
\newcommand{\outdelay}{\mathsf{delay}}

\newcommand{\yield}[1]{\texttt{yield}[#1]}
\newcommand{\window}{w}
\newcommand{\salgo}{\mathcal{E}}

\newcommand{\DS}{\mathsf{DS}}
\newcommand{\DSw}{\mathsf{DS}_w}
\newcommand{\dsnodes}{\operatorname{Nodes}(\mathsf{DS})}
\newcommand{\dsnodesw}{\operatorname{Nodes}(\mathsf{DS}_w)}
\newcommand{\n}{\mathsf{n}}
\newcommand{\dspos}{i}
\newcommand{\dslabels}{L}
\newcommand{\dsprod}{\operatorname{prod}}
\newcommand{\dsleft}{\operatorname{uleft}}
\newcommand{\dsright}{\operatorname{uright}}
\newcommand{\dsnull}{\bot}
\newcommand{\dsmaxstart}{\operatorname{max-start}}

\newcommand{\dsextend}{\mathtt{extend}}
\newcommand{\dsunion}{\mathtt{union}}

\newcommand{\dssem}[1]{\sem{#1}}
\newcommand{\dssemprod}[1]{\sem{#1}_{\dsprod}}

\newcommand{\valop}{\oplus}
\newcommand{\bigvalop}{\bigoplus}

\newcommand{\evalprob}{\textsl{EvalPCEA}}

\newcommand{\htablesym}{\mathsf{H}}
\newcommand{\htable}[3]{\htablesym[#1,#2,#3]}
\newcommand{\nset}{\mathsf{N}}
\newcommand{\nsetq}[1]{\nset_{#1}}
\newcommand{\ipos}{i}

\newcommand{\blank}{\hspace{0cm}}

\newcommand{\cev}[1]{\reflectbox{\ensuremath{\vec{\reflectbox{\ensuremath{#1}}}}}}

\newcommand{\qtree}{\tau_Q}
\newcommand{\cqtree}{\tau_Q^c}

\newcommand{\SJ}[1]{\operatorname{SJ}_{Q}}
\newcommand{\VSJ}[1]{\operatorname{xSJ}_{Q}}

\newcommand{\SchemaEX}{\Schema_{0}}
\newcommand{\StreamEX}{\Stream_{0}}
\newcommand{\relindex}[3]{\underbrace{#1(#3)}_{\texttt{\scriptsize{#2}}}}

\newcommand{\CHCEAEX}{\CHCEA_{0}}
\newcommand{\PFAEX}{\cP_{0}}
\newcommand{\PCEAEX}{\ACH_{0}}

\newcommand{\DBEX}{D_{0}}
\newcommand{\QEXZERO}{Q_{0}}
\newcommand{\QEXONE}{Q_{1}}

\newcommand{\qt}{q-tree\xspace}
\newcommand{\cqt}{compact q-tree\xspace}

\title{Complex event recognition meets hierarchical conjunctive queries}

\usepackage{authblk}

\author[1]{Dante Pinto}
\author[1]{Cristian Riveros}
\affil[1]{Pontificia Universidad Católica de Chile, \texttt{drpinto1@uc.cl}, \texttt{cristian.riveros@uc.cl}}

\date{} 

\begin{document}
	
	\maketitle
	
\begin{abstract}

Hierarchical conjunctive queries (HCQ) are a subclass of conjunctive queries (CQ) with robust algorithmic properties. Among others, Berkholz, Keppeler, and Schweikardt have shown that HCQ is the subclass of CQ (without projection) that admits dynamic query evaluation with constant update time and constant delay enumeration. On a different but related setting stands Complex Event Recognition (CER), a prominent technology for evaluating sequence patterns over streams. Since one can interpret a data stream as an unbounded sequence of inserts in dynamic query evaluation, it is natural to ask to which extent CER can take advantage of HCQ to find a robust class of queries that can be evaluated efficiently.

In this paper, we search to combine HCQ with sequence patterns to find a class of CER queries that can get the best of both worlds. To reach this goal, we propose a class of complex event automata model called Parallelized Complex Event Automata (PCEA) for evaluating CER queries with correlation (i.e., joins) over streams. This model allows us to express sequence patterns and compare values among tuples, but it also allows us to express conjunctions by incorporating a novel form of non-determinism that we call parallelization. 
We show that for every HCQ (under bag semantics), we can construct an equivalent PCEA. Further, we show that HCQ is the biggest class of acyclic CQ that this automata model can define. Then, PCEA stands as a sweet spot that precisely \changed{expresses} HCQ (i.e., among acyclic CQ) and extends them with sequence patterns. Finally, we show that PCEA also inherits the good algorithmic properties of HCQ by presenting a streaming evaluation algorithm under sliding windows with logarithmic update time and output-linear delay for the class of PCEA with equality predicates.

\end{abstract} 	
	
	\section{Introduction}\label{sec:introduction}

\emph{Hierarchical Conjunctive Queries}~\cite{DalviS07a} (HCQ) are a subclass of Conjunctive Queries (CQ) with good algorithmic properties for dynamic query evaluation~\cite{chirkova2012materialized,IdrisUV17}. In this scenario, users want to continuously evaluate a CQ over a database that receives insertion, updates, or deletes of tuples, and to efficiently retrieve the output after each modification. A landmark result by Berkholz, Keppeler, and Schweikardt~\cite{CQUpdates} shows that HCQ are the subfragment among CQ for dynamic query evaluation. Specifically, they show one can evaluate every HCQ with constant update time and constant-delay enumeration. Furthermore, they show that HCQ are the only class of full CQ (i.e., CQ without projection) with such guarantees, namely, under fined-grained complexity assumptions, a full CQ can be evaluated with constant update time and constant delay enumeration if, and only if, the query is hierarchical. Therefore, HCQ stand as the fragment for efficient evaluation under a dynamic scenario (see also~\cite{IdrisUV17}).

Data stream processing is another dynamic scenario where we want to evaluate queries continuously but now over an unbounded sequence of tuples (i.e., a data stream). \emph{Complex Event Recognition} (CER) is one such technology for processing information flow~\cite{GiatrakosAADG20, cugola2012processing}. CER systems read high-velocity streams of data, called events, and evaluate expressive patterns for detecting complex events, a subset of relevant events that witness a critical case for a user. A singular aspect of CER compared to other frameworks is that the order of the stream's data matters, reflecting the temporal order of events in reality (see~\cite{UgarteV18}). For this reason, sequencing operators are first citizens on CER query languages, which one combines with other operators, like filtering, disjunction, and correlation (i.e., joins), among others~\cite{ArtikisMUVW17}. 

Similar to dynamic query evaluation, this work aims to find a class of CER query languages with efficient streaming query evaluation. Our strategy to pursue this goal is simple but effective: we use HCQ as a starting point to guide our search for CER query languages with good algorithmic properties. Since one can interpret a data stream as an unbounded sequence of inserts in dynamic query evaluation, we want to extend HCQ with sequencing while maintaining efficient evaluation. We plan this strategy from an algorithmic point of view. Instead of studying which CER query language fragments have such properties, we look for automata models that can express HCQ. By finding such a model, we can later design our CER query language to express these queries~\cite{GrezRUV21}.

With this goal and strategy in mind, we start from the proposal of Chain Complex Event Automata (CCEA), an automata model for CER expressing sequencing queries with correlation, but that cannot express simple HCQ~\cite{grez-chain}. We extend this model with a new sort of non-deterministic power that we call \emph{parallelization}. This feature allows us to run several parallel executions that start independently and to gather them together when reading new data items. We define the class of \emph{Parallelized Complex Event Automata} (\acrocea), the extension of CCEA with parallelization. As an extension, PCEA can express patterns with sequencing, disjunction, iteration, and correlation but also allows conjunction. In particular, we can show that PCEA can express an acyclic CQ $Q$ if, and only if, $Q$ is hierarchical. Then, PCEA is a sweet spot that precisely \changed{expresses} HCQ (i.e., among acyclic CQ) and extends them with sequencing and other operations. Moreover, we show that PCEA inherits the good algorithmic properties of HCQ by presenting a streaming evaluation algorithm under sliding windows, reaching our desired goal. 

\paragraph{Contributions} The technical contributions and outline of the paper are the following.
 
 In Section~\ref{sec:preliminaries}, we provide some basic definitions plus recalling the definition of CCEA. 

 In Section~\ref{sec:pcea}, we introduce the concept of parallelization for standard non-deterministic NFA, called PFA, and study their properties. We show that PFA can be determinized in exponential time (similar to NFA) (Proposition~\ref{prop:pfa-are-regular}). We then apply this notion to CER and define the class of PCEA, showing that it is strictly more expressive than CCEA (Proposition~\ref{prop:pcea-expressive}). 

 Section~\ref{sec:hierarchical} compares PCEA with HCQ under bag semantics. Given that PCEA runs over streams and HCQ over relational databases, we must revisit the semantics of HCQ and formalize in which sense an HCQ and a PCEA define the same query. We show that under such comparison, every HCQ $Q$ under bag semantics can be expressed by a PCEA with equality predicates of exponential size in $|Q|$ and of \changed{quadratic} size if $Q$ does not have self joins (Theorem~\ref{theo:hierarchical-if}). Furthermore, if $Q$ is acyclic but not hierarchical, then $Q$ cannot be defined by any PCEA (Theorem~\ref{theo:hierarchical-onlyif}).  

 In Section~\ref{sec:algorithm}, we study the evaluation of PCEA in a streaming scenario. Specifically, we present a streaming evaluation algorithm under a sliding window with logarithmic update time and output-linear delay for the class of unambiguous PCEA with equality predicates~(Theorem~\ref{theo:algorithm}).

\paragraph{Related work} Dynamic query evaluation of HCQ and acyclic CQ has been studied in~\cite{CQUpdates,IdrisUV17,WangHDY23,KaraNOZ20}. This research line did not study HCQ or acyclic CQ in the presence of order predicates. \cite{WangY22,TziavelisGR21} studied CQ under comparisons (i.e., $\theta$-joins) but in a static setting (i.e., no updates). The closest work is~\cite{IdrisUVVL20}, which studied dynamic query evaluation of CQ with comparisons; however, this work did not study well-behaved classes of HCQ with comparisons, and, further, their algorithms have update time linear in the data.

Complex event recognition and, more generally, data stream processing have studied the evaluation of joins over streams~(see, e.g., \cite{xie2007survey, wu2006high, lin2015scalable}). To the best of our knowledge, no work in this research line optimizes queries focused on HCQ or provides guarantees regarding update time or enumeration delay in this setting. We base our work on~\cite{grez-chain}, which we will discuss extensively.%

	\section{Preliminaries}\label{sec:preliminaries}

\paragraph{Strings and NFA} A \emph{string} is a sequence of elements $\bar{s} = a_0 \ldots a_{n-1}$. For presentation purposes, we make no distinction between a \emph{sequence} or a string and, thus, we also write $\bar{s} = a_0, \ldots, a_{n-1}$ for denoting a string. We will denote strings using a bar and its $i$-th element by $\bar{s}[i] = a_i$. We use $|\bar{s}| = n$ for the length of $\bar{s}$ and $\{\bar{s}\} = \{a_0, \ldots, a_{n-1}\}$ to consider $\bar{s}$ as a set. Given two strings $\bar{s}$ and $\bar{s}'$, we write $\bar{s}\bar{s}'$ for the \emph{concatenation} of $\bar{s}$ followed by $\bar{s}'$.  Further, we say that $\bar{s}'$ is a \emph{prefix} of $\bar{s}$, written as $\bar{s'} \prefix \bar{s}$, if $|\bar{s}'| \leq |\bar{s}|$ and $\bar{s}'[i] = \bar{s}[i]$ for all $i < |\bar{s}'|$. Given a non-empty set $\Sigma$ we denote by $\Sigma^*$ the set of all strings from elements in~$\Sigma$, where $\epsilon \in \Sigma^*$ denotes the $0$-length string. For a function $f: \Sigma \rightarrow \Omega$ and $\bar{s} \in \Sigma^*$, we write $f(\bar{s}) = f(a_0) \ldots f(a_{n-1})$ to denote the point-wise application of $f$ over $\bar{s}$.  

A \emph{Non-deterministic Finite Automaton} (NFA) is a tuple $\cA = (Q, \Sigma, \Delta, I, F)$ such that $Q$ is a finite set of states, $\Sigma$ is a finite alphabet, $\Delta \subseteq Q \times \Sigma \times Q$ is the transition relation, and $I$ and $F$ are the set of initial and final states, respectively. A \emph{run} of $\cA$ over a string $\bar{s} = a_0 \ldots a_{n-1} \in \Sigma^*$ is a non-empty sequence $p_0 \ldots p_{n}$ such that $p_0 \in I$, and $(p_i, a_i, p_{i+1}) \in \Delta$ for every $i < n$. We say that $\cA$ \emph{accepts} a string $\bar{s} \in \Sigma^*$ iff there exists such a run of $\cA$ over $\bar{s}$ such that $p_{n} \in F$. We define the language $\cL(\cA)\subseteq \Sigma^*$ of all strings accepted by $\cA$. Finally we say that $\cA$ is a \changed{\emph{Deterministic Finite Automaton}} (DFA) iff $\Delta$ is given as a partial function $\Delta: Q \times \Sigma \rightarrow Q$ and $|I| = 1$.

\paragraph{Schemas, tuples, and streams} Fix a set $\Data$ of \emph{data values}. A \emph{relational schema} $\Schema$ (or just schema) is a pair $(\T, \arity)$ where $\T$ are the \emph{relation names} and $\arity: \T \rightarrow \bbN$ maps each name to a number, that is, its \emph{arity}. 
An \emph{$R$-tuple} of $\Schema$ (or just a tuple) is an object $R(a_0, \ldots, a_{k-1})$ such that  $R \in \T$, each $a_i \in \Data$, and $k = \arity(R)$. We will write $R(\bar{a})$ to denote a tuple with values $\bar{a}$.
We denote by $\tuples[\Schema]$ the set of all $R$-tuples of all $R \in \T$. We define the size of a tuple $R(\bar{a})$ as $|R(\bar{a})| = \sum_{i=0}^{k-1} |\bar{a}[i]|$ with $k = \arity(R)$ where $|\bar{a}[i]|$ is the size of the data value $\bar{a}[i] \in \Data$, which depends on the domain. %

A \emph{stream} $\Stream$ over $\Schema$ is an infinite sequence of tuples $\Stream = t_0 t_1 t_2 \ldots$ such that $t_i \in \tuples[\Schema]$ for every $i \geq 0$. For a running example, consider the schema $\SchemaEX$  with relation names $\T = \{R, S, T\}$, $\arity(R) = \arity(S) = 2$ and $\arity(T) = 1$. A stream $\StreamEX$ over $\SchemaEX$ could be the following:
\[
\StreamEX \ := \ 
\relindex{S}{0}{2,11} \
\relindex{T}{1}{2} \ 
\relindex{R}{2}{1,10} \ 
\relindex{S}{3}{2,11} \ 
\relindex{T}{4}{1} \ 
\relindex{R}{5}{2,11} \ 
\relindex{S}{6}{4,13}\ 
\relindex{T}{7}{1} \ \ldots
\]
where we add an index (i.e., the position) below each tuple (for simplification, we use $\Data = \bbN$).

\paragraph{Predicates} For a fix $k$, a \emph{$k$-predicate} $P$ is a subset of $\tuples[\Schema]^k$. Further, we say that $\bar{t} = (t_1, \ldots, t_k)$ \emph{satisfies} $P$ iff $\bar{t} \in P$. We say that $P$ is \emph{unary} if $k =1$ and \emph{binary} if $k=2$.  In the following, we denote any class of unary or binary predicates by $\un$ or $\bin$, respectively. 

Although we define our automata models for any class of unary and binary predicates, the following two predicate classes will be relevant for algorithmic purposes (see Section~\ref{sec:hierarchical} and~\ref{sec:algorithm}). Let $\sigma$ be a schema. We denote by $\uncq$ the class of all unary predicates $U$ \changed{such that}, for every $t \in \tuples[\Schema]$, one can decide in linear time over $|t|$ whether $t$ satisfies $U$ or not. \changed{In addition}, we denote by $\bincq$ the class of all equality predicates defined as follow: a binary predicate $B$ is an \emph{equality predicate} iff there exist partial functions $\leftbinfunc{B}$ and $\rightbinfunc{B}$ over $\tuples[\Schema]$ such that, for every $t_1,t_2 \in \tuples[\Schema]$, $(t_1, t_2) \in B$ iff $\leftbinfunc{B}(t_1)$ and $\rightbinfunc{B}(t_2)$ are defined and $\leftbinfunc{B}(t_1) = \rightbinfunc{B}(t_2)$. Further, we require that one can compute $\leftbinfunc{B}(t_1)$ and $\rightbinfunc{B}(t_2)$ in linear time over $|t_1|$ and $|t_2|$, respectively. 
For example, recall our schema $\SchemaEX$ and consider the binary predicate $(Tx, Sxy) = \{(T(a), S(a,b)) \mid a,b \in \Data\}$. Then by using the functions $\leftbinfunc{B}(T(a)) = a$ and $\rightbinfunc{B}(S(a,b)) = a$, one can check that $(Tx, Sxy)$ is an equality predicate. 

\changed{Note that $\bincq$ is a more general class of equality predicates compared with the ones used in~\cite{grez-chain}, that will serve in our automata models for comparing tuples by ``equality'' in different subsets of attributes.} We take here a more semantic presentation, where the equality comparison between tuples is directly given by the functions $\leftbinfunc{B}$ and $\rightbinfunc{B}$ and not symbolically by some formula. 

\paragraph{Chain complex event automata} A \emph{Chain Complex Event Automaton} (\chain)~\cite{grez-chain} is a tuple
$
\CHCEA \ = \ (Q, \un, \bin, \Omega, \Delta, I, F)
$ 
where $Q$ is a finite set of states, $\un$ is a set of unary predicates, $\bin$ is a set of binary predicates, $\Omega$ is a finite set of labels, $I: Q \rightarrow \un \times (2^\Omega \setminus \{\emptyset\})$ is a partial initial function, $F \subseteq Q$ is the set of final states, and $\Delta$ is a finite transition \changed{relation} of the form:
$
\Delta \subseteq Q \times \un \times \bin \times (2^\Omega \setminus \{\emptyset\}) \times Q.
$
Let $\Stream = t_0 t_1 \ldots$ be a stream. 
A \emph{configuration} of $\CHCEA$ over $\Stream$ is a tuple $(p, i, L) \in Q \times \nat \times (2^\Omega \setminus \{\emptyset\})$, representing that \changed{the automaton $\CHCEA$, is at state $p$ after having read and marked $t_i$ with the set of labels $L$.
For $\ell \in \Omega$, we say that $(p, i, L)$ \emph{marked} position $i$ with $\ell$ iff $\ell \in L$.}
Given a position $n\in \bbN$, we say that a configuration is \emph{accepting at position} $n$ iff it is of the form $(p, n, L)$ and~$p \in F$. 
Then a \emph{run} $\rho$ of $\CHCEA$ over $\Stream$ is a sequence of configurations:
\[
\rho \ := \ (p_0, i_0, L_0), (p_1, i_1, L_1), \ldots, (p_n, i_n, L_n) 
\]
such that $i_0 < i_1 < \ldots < i_n$, $I(p_0) = (U, L_0)$ is defined and $t_{i_0} \in U$, and there exists a transition $(p_{j-1}, U_j, B_j, L_j, p_j) \in \Delta$ such that $t_{i_j} \in U_j$ and $(t_{i_{j-1}}, t_{i_j}) \in B_j$ for every $j \in [1, n]$. Intuitively, a run of a \chain is a subsequence of the stream that can follow a path of transitions, where each transition \changed{checks} a local condition (i.e., the unary predicate $U_j$) and a join condition (i.e., the binary predicate $B_j$) with the previous tuple. For the first tuple, a \chain can only check a local condition (i.e., there is no previous tuple). 

Given a run $\rho$ like above, we define its \emph{valuation} $\nu_\rho: \Omega \rightarrow 2^{\bbN}$ such that $\nu_\rho(\ell)$ is the set consisting of all positions in $\rho$ marked by $\ell$, formally, $\nu_\rho(\ell) = \{ i_j \mid\, j \leq n \wedge \ell \in L_j\}$.
Further, given a position $n\in \bbN$, we say that $\rho$ is an \emph{accepting run at position} $n$ iff $(p_n, i_n, L_n)$ is an accepting configuration at $n$.
Then the \emph{output} of $\CHCEA$ over $\Stream$ at position $n$ is defined~as:
\[
\sem{\CHCEA}_n(\Stream) \ = \  \{\nu_\rho \mid \text{$\rho$ is an accepting run at position $n$ of $\CHCEA$ over $S$}\}.
\]
\begin{example}\label{ex:chain}
	Below, we show an example of a \chain over the schema $\SchemaEX$ with $\Omega = \{\bullet\}$:
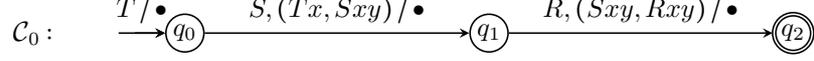
	\begin{center}
	\begin{tikzpicture}[>=stealth, 
		semithick, 
		auto,
		initial text= {},
		initial distance= {4mm},
		accepting distance= {3mm},
		mystate/.style={state, inner sep=0pt, minimum size=5mm}]
		
		\node at (-2, 0) {$\CHCEAEX:$};

		\node (E) at (-1, 0) {};
		
		\node [mystate, initial left] (T) at (0, 0) {$q_0$};
		
		\node [mystate] (S) at (4, 0) {$q_1$};
		
		\node [mystate, accepting] (R) at (8, 0) {$q_2$};
		
		\draw[->] (E) edge node {$T \, / \, \bullet$} (T);
		\draw[->] (T) edge node {$S,(Tx,Sxy) \, / \,  \bullet$} (S);
		\draw[->] (S) edge node {$R,(Sxy,Rxy) \, / \,   \bullet$} (R);
	\end{tikzpicture}	
	\end{center}
	We use $T$ to denote the predicate $T = \{T(a) \mid a \in \Data\}$ and similar for $S$ and $R$. Further, we use $(Tx, Sxy)$ and $(Sxy,Rxy)$ to denote equality predicates as defined above. An accepting run of $\CHCEAEX$ over $\StreamEX$ is $\rho = (q_0, 1, \{\bullet\}), (q_1, 3, \{\bullet\}), (q_2, 5, \{\bullet\})$ which produces the valuation $\nu_\rho = \{\bullet \mapsto \{1,3,5\}\}$ that represents the subsequence $T(2), S(2,11), R(2,11)$ of $\StreamEX$.
	Intuitively, $\CHCEAEX$ defines all subsequences of the form $T(a), S(a,b), R(a,b)$ for every~$a, b \in \Data$.
\end{example}

Note that the definition of \chain above differs from~\cite{grez-chain} to fit our purpose better. Specifically, we use a set of labels $\Omega$ to annotate positions in the streams and define valuations in the same spirit as the model of \emph{annotated automata} used in \cite{AmarilliJMR22,MunozR23}. One can see this extension as a generalization to the model in~\cite{grez-chain}, where $|\Omega| = 1$. This extension will be helpful to enrich the outputs of our models for comparing them with hierarchical conjunctive queries with self-joins (see Section~\ref{sec:hierarchical}). 

\paragraph{Computational model} For our algorithms, we assume the computational model of Random Access Machines (RAM) with uniform cost measure, and addition as it basic operation~\cite{aho1974design, GradjeanJ22}. This RAM has read-only registers for the input, read-writes registers for the work, and write-only registers for the output. This computation model \changed{is a standard} assumption in the literature~\cite{CQUpdates,BucchiGQRV22}. 	
	\section{Parallelized complex event automata}\label{sec:pcea}

This section presents our automata model for specifying CER queries with conjunction called \anamelong (\acrocea), which strictly generalized \chain by adding a new feature called \emph{parallelization}. For the sake of presentation, we first formalize the notion of parallelization for NFA to extend the idea to \chain. Before this, we need the notation of labeled trees that will be useful for our definitions and proofs.

\paragraph{Labeled trees} As it is common in the area~\cite{Neven02}, we define (unordered) \emph{trees} as a finite set of strings $t \subseteq \bbN^\ast$ that satisfies two conditions: (1) $t$ contains the empty string, (i.e., $\varepsilon \in t$), and (2) $t$ is a \emph{prefix-closed set}, namely, if $a_1 ...a_{n} \in t$, then $a_1 ...a_{j} \in t$ for every $j < n$. We will refer to the string of $t$ as \emph{nodes}, and the \emph{root} of a tree, $\rt{t}$, will be the empty string $\varepsilon$.

Let $\bar{u}, \bar{v} \in t$ be nodes. The \emph{depth} of $\bar{u}$ will be given by its length $\depth{t}{\bar{u}} = |\bar{u}|$. We say that $\bar{u}$ is the \emph{parent} of $\bar{v}$ and write $\parent{t}{\bar{v}}=\bar{u}$ if $\bar{v} = \bar{u} \cdot n$ for some $n \in \bbN$. \changed{Likewise}, we say that $\bar{v}$ is a \emph{child} of $\bar{u}$ if $\bar{u}$ is the parent of $\bar{v}$ and define $\children{t}{\bar{u}} = \{\bar{v}\in t \mid \parent{t}{\bar{v}} = \bar{u}\}$.
Similarly, we define the \emph{descendants} of $\bar{u}$ as $\desc{t}{\bar{u}} = \{ \bar{v} \in t \mid \bar{u} \prefix \bar{v} \}$ and the \emph{ancestors} as $\ancst{t}{\bar{u}} = \{ \bar{v} \in t \mid \bar{v} \prefix \bar{u} \}$; note that $\bar{u} \in \desc{t}{\bar{u}}$ and $\bar{u} \in \ancst{t}{\bar{u}}$. A node $\bar{u}$ is a \textit{leaf} of $t$ if $\desc{t}{\bar{u}} = \{\bar{u}\}$, and an \emph{inner node} if it is not a leaf node. We define the \emph{set of leaves} of $\bar{u}$ as $\leaves{t}{\bar{u}} = \{\bar{v} \in \desc{t}{\bar{u}} \mid \bar{v} \text{ is a leaf node} \}$.

A \emph{labeled tree} $\tau$ is a function $\tau\colon t \rightarrow L$ where $t$ is a tree and $L$ is any finite set of labels. We use $\dom(\tau)$ to denote the underlying tree structure $t$ of $\tau$. Given that $\tau$ is a function, we can write $\tau(\bar{u})$ to denote the label of node $\bar{u} \in \dom(\tau)$. To simplify the notation, we extend all the definitions above for a tree $t$ to labeled tree $\tau$, changing $t$ by $\dom(\tau)$. For example, we write $\bar{u} \in \tau$ to refer to $\bar{u} \in \dom(\tau)$, or $\parent{\tau}{\bar{u}}$ to refer to $\parent{\dom(\tau)}{\bar{u}}$.  
Finally, we say that two labeled trees $\tau$ and $\tau'$ are \emph{isomorphic} if there exists a bijection $f\colon \dom(\tau) \rightarrow \dom(\tau')$ such that $\bar{u} \prefix \bar{v}$ iff $f(\bar{u}) \prefix f(\bar{v})$ and $\tau(\bar{u}) = \tau'(f(\bar{u}))$ for every $\bar{u}, \bar{v} \in \dom(\tau)$. We will usually say that $\tau$ and $\tau'$ are equal, meaning they are isomorphic. 

\paragraph{Parallelized finite automata} 
A \emph{Parallelized Finite Automaton} (\aacro) is a tuple $\cP = (Q, \Sigma, \Delta, I, F)$ where $Q$ is a finite set of states, $\Sigma$ is a finite alphabet, $I, F \subseteq Q$ are the sets of initial and accepting states, respectively, and $\Delta \subseteq 2^Q \times \Sigma \times Q$ is the transition relation.
\changed{We define the size of $\cP$ as $|\cP| = |Q| + \sum_{(P,a,q)} (|P|+1)$, namely, the number of states plus the size of encoding the transitions.}

A \emph{run tree} of a \aacro $\cP$ over a string $\bar{s}=a_1\ldots a_n \in \Sigmastar$ is a labeled tree \changed{$\tau: t \rightarrow \dom(\tau)$} such that $\depth{\tau}{\bar{u}} = n$ for every leaf $\bar{u} \in \tau$; \changed{in other words, every node of $\tau$ is labeled by a state of $\cP$} and all branches have the same length $n$. In addition, $\tau$ must satisfy the following two conditions: (1) every leaf node $\bar{u}$ of $t$ is  labeled by an initial state (i.e., $\tau(\bar{u}) \in I$) and (2) for every inner node $\bar{v}$ at depth $i$ (i.e., $\depth{\tau}{\bar{v}} = i$) there must be a transition $(P, a_{n-i}, q) \in \Delta$ such that $\tau(\bar{v}) = q$, $|\children{\tau}{\bar{v}}| = |P|$ and $P = \{\tau(\bar{u}) \mid \bar{u} \in \children{\tau}{\bar{v}}\}$, that is, 
\changed{children have different labels and $P$ is the set of labels in the children of $\bar{v}$.}
We say that $\tau$ is an accepting run of $\cP$ over $\bar{s}$ iff $\tau$ is a run of $\cP$ over $\bar{s}$ and $\tau(\varepsilon) \in F$ (recall that $\varepsilon = \rt{\tau}$).   
We say that $\cP$ accepts a string $\bar{s} \in \Sigmastar$ if there is an accepting run of $\cP$ over $\bar{s}$ and we define the language recognized by $\cP$, $\cL(\cP)$, as the set of strings that $\cP$ accepts.

\begin{example}\label{ex:pfa}
In Figure~\ref{fig:ex-pfa-pcea} (left), we show the example of a \aacro $\PFAEX$ over the alphabet $\Sigma = \{T,S,R\}$. Intuitively, the upper part (i.e., $p_0, p_1$) looks for a symbol $T$, the lower part (i.e., $p_2, p_3$) for a symbol $S$, and both runs join together in $p_4$ when they see a symbol $R$. Then, $\PFAEX$ defines all strings that contain symbols $T$ and $S$ (in any order) before a symbol $R$. 
\end{example}

One can see that \aacro is a generalization of an NFA. Indeed, NFA is a special case of an \aacro where each run tree $\tau$ is a line.  
Nevertheless, \aacro do not add expressive power to NFA, given that \aacro is another model for recognizing regular languages, as the next result shows.
\begin{proposition}\label{prop:pfa-are-regular}
	For every \aacro $\cP$ with $n$ states there exists a DFA $\cA$ with at most $2^n$ states such that $\cL(\cP) = \cL(\cA)$. In particular, all languages defined by \aacro are regular. 
\end{proposition}

\changed{Intuitively, one could interpret a \aacro as an \emph{Alternating Finite Automaton} (AFA)~\cite{chandra1981alternation} that runs backwards over the string (however, they still process the string in a forward direction).} It was shown in~\cite[Theorem 5.2~and~5.3]{chandra1981alternation} that for every AFA that defines a language $L$ with $n$ states, there exists an equivalent DFA with $2^{2^n}$ states in the worst case that recognizes $L$. Nevertheless, they argued that the reverse language $L^R = \{a_1a_2 \ldots a_n \in \Sigmastar\mid a_n \ldots a_2 a_1 \in L\}$ can always be accepted by a DFA with at most $2^n$ states. Then, one can see Proposition~\ref{prop:pfa-are-regular} as a consequence of reversing an alternating automaton. Despite this connection, we use here \aacro as a proper automata model, which was not studied or used in~\cite{chandra1981alternation}. Another related proposal is the \emph{parallel finite automata} model presented in~\cite{StottsP94}. Indeed, one can consider \aacro as a restricted case of this model, although it was not studied in~\cite{StottsP94}. For this reason, we decided to name the \aacro model with the same acronym but a slightly different name as in~\cite{StottsP94}.

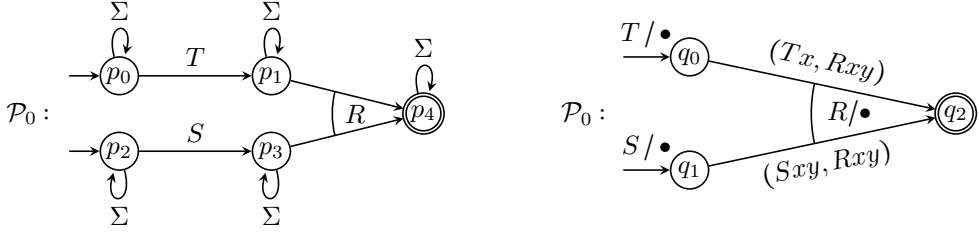
\begin{figure}[t]
	\centering
	\begin{tikzpicture}[>=stealth, 
		semithick, 
		auto,
		initial text= {},
		initial distance= {4mm},
		accepting distance= {3mm},
		mystate/.style={state, inner sep=0pt, minimum size=5mm}]
		
		\node at (-1.2,-0.5) {$\PFAEX:$};
		
		\node [mystate, initial left] (0) at (0, 0) {$p_0$};
		\node [mystate] (1) at (2, 0) {$p_1$};
		
		\node [mystate, initial left] (2) at (0, -1) {$p_2$};
		\node [mystate] (3) at (2, -1) {$p_3$};
		
		\node [mystate,accepting] (4) at (4, -0.5) {$p_4$};
		
		\draw[->] (0) edge[loop above, pos=0.5] node {$\Sigma$} (0);
		\draw[->] (1) edge[loop above, pos=0.5] node {$\Sigma$} (1);
		\draw[->] (2) edge[loop below, pos=0.5] node {$\Sigma$} (2);
		\draw[->] (3) edge[loop below, pos=0.5] node {$\Sigma$} (3);
		\draw[->] (4) edge[loop above, pos=0.5] node {$\Sigma$} (4);
		
		\draw[->] (0) edge node {$T$} (1);
		\draw[->] (2) edge node {$S$} (3);
		\draw[->] (1) edge (4);
		\draw[->] (3) edge (4);
		
		\pic [draw, -, angle radius=1.2cm] {angle = 1--4--3};
		\node at ($(4)+(-0.9,0)$) {$R$};		
		
		\begin{scope}[xshift=7.5cm]
			\node at (-1.4, -0.5) {$\PCEAEX:$};

			\node (ET) at (-1, 0.25) {};
			\node [mystate] (T) at (0, 0.25) {$q_0$};
			
			\node (ES) at (-1, -1.25) {};
			\node [mystate] (S) at (0, -1.25) {$q_1$};
			
			\node [mystate, accepting] (R) at (3.5, -0.5) {$q_2$};
			
			\draw[->] (ET) edge node {$T \, / \, \bullet$} (T);
			\draw[->] (ES) edge node {$S \, / \, \bullet$} (S);

			\draw[->] (T) edge node[sloped] {$(Tx,Rxy)$} (R);
			\draw[->] (S) edge node[sloped, swap] {$(Sxy,Rxy)$} (R);
			
			\pic [draw, -, angle radius=1.9cm] {angle = T--R--S};
			\node at ($(R)+(-1.4,0)$) {$R / \bullet$};
			
		\end{scope}	
	\end{tikzpicture}	
	\caption{On the left, an example of a \aacro and, on the right, an example of a \acrocea.}
	\label{fig:ex-pfa-pcea}
\end{figure}

\paragraph{Parallelized complex event automata} A \emph{Parallelized Complex Event Automaton} (\acrocea) is the extension of \chain with the idea of parallelization as in~\aacro. Specifically, a \acrocea is a tuple $\ACH = (Q, \un, \bin, \Omega, \Delta, F)$, where $Q$, $\un$, $\bin$, $\Omega$, and $F$ are the same as for \chain, and $\Delta$ is a finite transition relation of the form:
\[
\Delta \subseteq 2^Q \times \un \times \bin^Q \times (2^\Omega \setminus \{\emptyset\}) \times Q.
\]
where $\bin^Q$ are all partial functions $\binfunc: Q \rightarrow \bin$, that associate a state $q$ to a binary predicate~$\binfunc(q)$. We define the \emph{size of $\ACH$} as $|\ACH| = |Q| + \sum_{(P, U, \binfunc, L, q) \in \Delta} (|P|+|L|)$. Note that $\ACH$ does not define the initial function explicitly. As we will see, transitions of the form $(\emptyset, U, \binfunc, L, q)$ will play the role of the initial function on a run of $\ACH$.

Next, we extend the notion of a run from \chain to its parallelized version. Let $\Stream = t_0 t_1 \ldots$ be a stream.
A \emph{run tree} of $\ACH$ over $\Stream$ is now a labeled tree $\tau: t \rightarrow (Q \times \nat \times (2^\Omega \setminus \{\emptyset\}))$ 
where each node $\bar{u} \in \tau$ is labeled with a configuration $\tau(\bar{u}) = (q, i, L)$ such that, for every child $\bar{v} \in \children{\tau}{\bar{u}}$ with $\tau(\bar{v}) = (p, j, M)$, it holds that $j < i$. In other words, the positions of $\tau$-configurations increase towards the root of $\tau$, similar to the runs of a \chain.
In addition, $\bar{u}$ must satisfy the transition relation $\Delta$, that is, there must exist a transition $(P, U, \binfunc, L, q) \in \Delta$ such that (1) $t_i \in U$, (2) $|\children{\tau}{\bar{u}}| = |P|$ and $P = \{p \mid \exists \bar{v} \in \children{\tau}{\bar{u}}. \, \tau(\bar{v}) = (p, j, M)\}$,  and (3) for every $\bar{v} \in \children{\tau}{\bar{u}}$ with  $\tau(\bar{v}) = (p, j, M)$, $(t_j, t_i) \in \binfunc(p)$. Similar to \aacro, condition (2) forces that there exists a bijection between $P$ and the states at the children of $\bar{u}$. Instead, condition (3) forces that two consecutive configurations $(p, j, M)$ and $(q, i, L)$ must satisfy the binary predicate in $\binfunc(p)$ associated with~$p$. Notice that, if $\bar{u}$ is a leaf node in $\tau$, then it must hold that $P = \emptyset$ and condition (3) is trivially~satisfied. \changed{Also, note that we do not assume that all leaves are at the same depth.} 

Given a position $n\in \bbN$, we say that $\tau$ is an \emph{accepting run at position} $n$ iff the root configuration $\tau(\varepsilon)$ is accepting at position $n$. 
Further, we define the output of a run $\tau$ as the valuation $\nu_\tau: \Omega \rightarrow 2^{\bbN}$ such that $\nu_\tau(\ell) = \{ i \mid\, \exists \bar{u} \in \tau. \ \tau(\bar{u}) = (q, i, L) \wedge \ell \in L\}$ for every label $\ell \in \Omega$. Finally, \changed{the output of a \acrocea} $\ACH$ over $\Stream$ at the position $n$ is defined~as:
\[
\sem{\ACH}_n(\Stream) \ = \  \{\nu_\tau \mid \text{$\tau$ is an accepting run at position $n$ of $\ACH$ over $S$}\}.
\]
\begin{example}\label{ex:pcea}
	In Figure~\ref{fig:ex-pfa-pcea} (right), we show an example of a \acrocea $\PCEAEX$ over schema $\SchemaEX$ with $\Omega = \{\bullet\}$. We use the same notation as in  Example~\ref{ex:chain} to represent unary and equality predicates. If we run $\PCEAEX$ over $\StreamEX$, we have the following two run trees at position $5$:
	
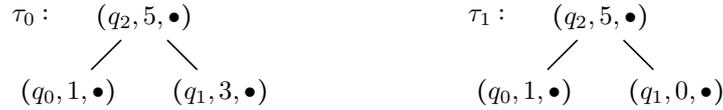
\begin{center}
	\begin{tikzpicture}[>=stealth, 
		semithick, 
		auto,
		initial text= {},
		initial distance= {4mm},
		accepting distance= {3mm},
		mystate/.style={state, inner sep=0pt, minimum size=5mm}]
		
		\begin{scope}[xshift=0cm]
			
			\node at (-1.5,0) {$\tau_0:$};
			
			\node (t2) at (0,0) {$(q_2,5,\bullet)$};
			
			\node (t0) at (-1,-1) {$(q_0,1,\bullet)$};
			\node (t1) at (1,-1) {$(q_1,3,\bullet)$};	
			
			\draw[-] (t0) edge (t2);
			\draw[-] (t1) edge (t2);

		\end{scope}
		
		\begin{scope}[xshift=6cm]
			
			\node at (-1.5,0) {$\tau_1:$};
					
			\node (t2) at (0,0) {$(q_2,5,\bullet)$};
			
			\node (t0) at (-1,-1) {$(q_0,1,\bullet)$};
			\node (t1) at (1,-1) {$(q_1,0,\bullet)$};	
			
			\draw[-] (t0) edge (t2);
			\draw[-] (t1) edge (t2);

		\end{scope}
	
	\end{tikzpicture}	
	\end{center}
	that produces the valuation $\nu_{\tau_0} = \{\bullet \mapsto \{1,3,5\}\}$ and $\nu_{\tau_1} = \{\bullet \mapsto \{0,1,5\}\}$ representing the subsequences $T(2), S(2,11), R(2,11)$ and $S(2,11), T(2), R(2,11)$ of $\StreamEX$, respectively. Note that the former is an output of $\CHCEAEX$ in Example~\ref{ex:chain}, but the latter is not. 
\end{example}

It is easy to see that every \chain is a \acrocea where every transition $(P, U, \binfunc, L, q) \in \Delta$ satisfies that $|P| \leq 1$. \changed{Additionally, the previous example} gives evidence that \acrocea is a strict generalization of \chain, namely, there exists no~\chain that can define $\PCEAEX$\changed{. Intuitively, since a \chain can only compare the current tuple to the last tuple, for a stream like $\Stream=R(a, b), T(a), S(a, b)$ it would be impossible to check conditions over the second attribute of tuples $R(a, b)$ and $S(a, b)$.}

\begin{proposition}\label{prop:pcea-expressive}
	\acrocea is strictly more expressive than \chain.
\end{proposition}

\paragraph{Unambiguous \acrocea} We end this section by introducing a subclass of \acrocea relevant to our algorithmic results. Let $\ACH$ be a \acrocea and $\tau$ a run of $\ACH$ over some stream. We say that $\tau$ is \emph{simple} iff for every two different nodes $\bar{u}, \bar{u}' \in \tau$ with $\tau(\bar{u}) = (q,i,L)$ and $\tau(\bar{u}') = (q', i', L')$, if $i = i'$, then $L \cap L' = \emptyset$. In other words, $\tau$ is simple if all positions of the valuation $\nu_\tau$ are uniquely represented in~$\tau$.
We say that $\ACH$ is \emph{unambiguous} if (1) every accepting run of $\ACH$ is simple and (2) \changed{for every stream $\Stream$ and accepting run $\tau'$ of $\ACH$ over $\Stream$ with valuation $\nu_{\tau}$, there is no other run $\tau$ of $\ACH$  with valuation $\nu_{\tau'}$ such that $\nu_\tau = \nu_{\tau'}$.} For example, the reader can check that $\PCEAEX$ is unambiguous.

Condition (2) of unambiguous \acrocea ensures that each output is witnessed by exactly one run. This condition is common in MSO enumeration~\cite{AmarilliBJM17,MunozR22} for a one-to-one correspondence between outputs and runs. \changed{Condition (1)} forces a correspondence between the size of the run and the size of the output it represents. As we will see, both conditions will be helpful for our evaluation algorithm, 
and satisfied by our translation of hierarchical conjunctive queries into PCEA in the next section.

	\section{Representing hierarchical conjunctive queries}\label{sec:hierarchical}

This section studies the connection between \acrocea and hierarchical conjunctive queries~(HCQ) over streams. For this purpose, we must first define the semantics of HCQ over streams and how to relate their expressiveness with \acrocea. We connect them by using a \emph{bag semantics} of~CQ. We start by introducing bags that will be useful throughout this section.

\paragraph{Bags} A bag (also called a multiset) is usually defined in the literature as a function that maps each element to its multiplicity (i.e., the number of times it appears). In this work, we use a different but equivalent representation of a bag where each element has its \emph{own identity}. This representation will be helpful in our context to deal with duplicates in the stream and define the semantics of hierarchical CQ in the case of self joins.

We define a \emph{bag} (with own identity) $B$ as a surjective function $B: I \rightarrow U$ where $I$ is a finite set of identifiers (i.e., the identity of each element) and $U$ is the underlying set of the bag. Given any bag $B$, we refer to these components as $I(B)$ and $U(B)$, respectively. For example, a bag $B = \bleft a, a, b \bright$ (where $a$ is repeated twice) can be represented with a surjective function $B_0 =\{0 \mapsto a, 1 \mapsto a, 2 \mapsto b\}$ where $I(B_0) = \{0, 1, 2\}$ and $U(B_0) = \{a,b\}$. In general, we will use the standard notation for bags $\bleft a_0, \ldots, a_{n-1} \bright$ to denote the bag $B$ whose identifiers are $I(B) = \{0, \ldots, n-1\}$ and $B(i) = a_i$ for each $i \in I(B)$. Note that if $B: I \rightarrow U$ is injective, then $B$ encodes a set (i.e., no repetitions). \changed{We write $a \in B$ if $B(i) = a$} for some $i \in I(B)$ and define the empty bag $\emptyset$ such that $I(\emptyset) = \emptyset$ and $U(\emptyset) = \emptyset$.

For a bag $B$ and an element $a$, we define the \emph{multiplicity} of $a$ in $B$ as $\mult{B}{a} = |\{i \mid B(i) = a\}|$. Then, we say that a bag $B'$ is \emph{contained} in $B$, denoted as $B' \subseteq B$, iff $\mult{B'}{a} \leq \mult{B}{a}$ for every $a$. We also say that two bags $B'$ and $B$ are \emph{equal}, and write $B = B'$, if $B' \subseteq B$ and $B \subseteq B'$. Note that two bags can be equal although the set of identifiers can be different (i.e., they are equal up to a renaming of the identifiers). Given a set $A$, we say that $B$ \emph{is a bag from elements of} $A$ (or just a bag of $A$) if $U(B) \subseteq A$.

\paragraph{Relational databases} Recall that $\Data$ is our set of data values and let $\Schema = (\T, \arity)$ be a schema. A \emph{relational database} $D$ (with duplicates) over $\Schema$ is a bag of $\tuples[\Schema]$. Given a relation name $R \in \T$, we write $R^D$ as the bag of $D$ containing only the $R$-tuples of $D$, formally, $I(R^D) = \{i \in I(D) \mid D(i) = R(\bar{a}) \text{ for some $\bar{a}$} \}$ and $R^D(i) = D(i)$ for every $i \in I(R^D)$. For example, consider again the schema $\SchemaEX$. Then a database $\DBEX$ over $\SchemaEX$ is the bag:
\[
\DBEX \ :=  \ \bleft \, S(2,11), T(2), R(1,10), S(2,11), T(1), R(2,11)\, \bright.
\]
Here, one can check that $T^{\DBEX} = \bleft T(2),  T(1) \bright$ and $S^{\DBEX} = \bleft S(2,11), S(2,11)\bright$.

\paragraph{Conjunctive queries} Fix a schema $\Schema = (\T, \arity)$ and a set of variables $\Var$ disjoint from $\Data$ (i.e., $\Var \cap \Data = \emptyset$). A \emph{Conjunctive Query} (CQ) over relational schema $\Schema$ is a syntactic structure of the form:
\[
	Q(\bar{x}) \ \leftarrow \ R_0(\bar{x}_0), \ldots, R_{m-1}(\bar{x}_{m-1})  \tag{\dag} \label{eq:cq}
\]
such that $Q$ is a relational name not in $\T$, $R_i \in \T$, $\bar{x}_i$ is a sequence of variables in $\Var$ and data values in $\Data$, and $|\bar{x}_i| = \arity(R_i)$ for every $i < m$. Further, $\bar{x}$ is a sequence of variables in $\bar{x}_0, \ldots, \bar{x}_{m-1}$. 
We will denote a CQ like (\ref{eq:cq}) by $Q$, where $Q(\bar{x})$ and $R_0(\bar{x}_0), \ldots, R_{m-1}(\bar{x}_{m-1})$ are called the \emph{head} and the \emph{body} of $Q$, respectively. 
Furthermore, we call each $R_i(\bar{x}_i)$ an \emph{atom} of $Q$. For example, the following are two conjunctive queries $Q_0$ and $Q_1$ over the schema $\Schema_0$:
\[
\QEXZERO(x,y) \leftarrow  T(x), \, S(x,y) , \, R(x,y) \ \ \ \ \ \ \ \ \ \ \ \ \QEXONE(x,y) \leftarrow T(x), \, R(x, y),\,  S(2,y), \, \,  T(x)
\]
Note that a query can repeat atoms. For this reason, we will regularly consider $Q$ as a bag of atoms, where $I(Q)$ are the positions of $Q$ and $U(Q)$ is the set of distinct atoms.
For instance, we can consider $\QEXONE$ above as a bag of atoms, where $I(Q_1) = \{0, 1, 2, 3\}$ (i.e., the position of the atoms) and $\QEXONE(0) = T(x)$, $\QEXONE(1) = R(x, y)$, $\QEXONE(2) = S(2, y)$, $\QEXONE(3) = T(x)$.
We say that a CQ $Q$ has \emph{self joins} if there are two atoms with the same relation name. We can see in the previous example that $\QEXONE$ has self  joins, while $\QEXZERO$ does not.

\paragraph{Homomorphisms and CQ bag semantics} Let $Q$ be a CQ, and $D$ be a database over the same schema $\Schema$. A \emph{homomorphism} is any function $h: \Var \cup \Data \rightarrow \Data$ such that $h(a) = a$ for every $a \in \Data$. 
We extend $h$ as a function from atoms to tuples such that $h(R(\bar{x})) := R(h(\bar{x}))$ for every atom $R(\bar{x})$. 
We say that $h$ is a homomorphism from $Q$ to $D$ if $h$ is a homomorphism and $h(R(\bar{x})) \in D$ for every atom $R(\bar{x})$ in $Q$. We denote by $\Hom(Q, D)$ the set of all homomorphisms from $Q$ to $D$.

To define the bag semantics of CQ, we need a more refined notion of homomorphism that specifies the correspondence between atoms in $Q$ and tuples in $D$. Formally,  a \emph{tuple-homomorphism} from $Q$ to $D$ (or \emph{t-homomorphism} for short) is a function $\eta: I(Q) \rightarrow I(D)$ such that there exists a homomorphism $h_\eta$ from $Q$ to $D$ satisfying that $h_\eta(Q(i)) = D(\eta(i))$ for every $i \in I(Q)$. 
For example, consider again $\QEXZERO$ and $\DBEX$ above, then $\eta_0 = \{0 \mapsto 1, 1 \mapsto 3, 2, \mapsto 5\}$ and $\eta_1 = \{0\mapsto 1, 1 \mapsto 0, 2 \mapsto 5\}$ are two t-homomorphism from $\QEXZERO$ to $\DBEX$. 

Intuitively, a t-homomorphism is like a homomorphism, but it additionally specifies the correspondence between atoms (i.e., $I(Q)$) and tuples (i.e.,  $I(D)$) in the underlying bags. One can easily check that if $\eta$ is a t-homomorphism, then $h_\eta$ (restricted to the variables of $Q$) is unique. For this reason, we usually say that $h_\eta$ is the homomorphism associated to~$\eta$. Note that the converse does not hold: for $h$ from $Q$ to $D$, there can be several t-homomorphisms $\eta$ such that $h = h_\eta$.

Let $Q(\bar{x})$ be the head of $Q$. 
We define the output of a CQ $Q$ over a database $D$ as:
\[ 
\sem{Q}(D) \ = \ \bleft Q(h_\eta(\bar{x})) \,\mid\, \text{$\eta$ is a  t-homomorphism from $Q$ to $D$} \bright.
\]
Note that the result is another relation where each %
$Q(h_\eta(\bar{x}))$ is witnessed by a t-homomorphism from $Q$ to $D$. In other words, there is a one-to-one correspondence between tuples in $\sem{Q}(D)$ and t-homomorphisms from $Q$ to $D$.

\paragraph{Discussion} In the literature, homomorphisms are usually used to define the set semantics of a CQ $Q$ over a database $D$.
They are helpful for set semantics but ``inconvenient'' for bag semantics since it does not specify the correspondence between atoms and tuples; namely, they only witness the existence of such correspondence.
In~\cite{chaudhuri1993optimization}, Chaudhuri and Vardi introduced the bag semantics of CQ by using homomorphisms, which we recall next. Let $Q$ be a CQ like (\ref{eq:cq}) and $D$ a database over the same schema $\sigma$, and let $h \in \Hom(Q, D)$. We define the \emph{multiplicity} of $h$ with respect to $Q$ and $D$ by: 
\[
	\mult{Q,D}{h} \ = \  \prod_{i = 0}^{m-1} \mult{D}{h(R_i(\bar{x}_i))}
\]
Chaudhuri and Vardi defined the bag semantics $\semaux{Q}$ of $Q$ over $D$ as the bag $\semaux{Q}(D)$ such that each tuple $Q(\bar{a})$ has multiplicity equal to:
\[
	\mult{\semaux{Q}(D)}{Q(\bar{a})} \ = \ \sum_{h \in \Hom(Q,D)\, \colon\, h(\bar{x}) = \bar{a}} \mult{Q,D}{h}
\]

In Appendix~\ref{sec:app-hierarchical}, we prove that for every CQ $Q$ and database $D$ it holds that $\sem{Q}(D) = \semaux{Q}(D)$, namely, the bag semantics introduced here (i.e., with t-homomorphisms) is equivalent to the standard bag semantics of CQ.
The main difference is that the standard bag semantics of CQ are defined in terms of homomorphisms and multiplicities, and there is no direct correspondence between outputs and homomorphisms. For this reason, we redefine the bag semantics of CQ in terms of t-homomorphism that will connect the outputs of CQ with the outputs of \acrocea over~streams. 

\paragraph{CQ over streams} Now, we define the semantics of CQ over streams, formalizing its comparison with queries in complex event recognition.
For this purpose, we must show how to interpret streams as databases and encode CQ's outputs as valuations. 
Fix a schema $\Schema$ and a stream $\Stream = t_0 t_1 \cdots$ over $\sigma$. Given a position $n \in \bbN$, we define the database of $\Stream$ at position $n$ as the $\sigma$-database $D_n[\Stream] = \bleft t_0, t_1, \ldots, t_n \bright$. For example, $D_5[\StreamEX] = \DBEX$. One can interpret here that $\Stream$ is a sequence of inserts, and then $D_n[\Stream]$ is the database version at position~$n$. 
Since $D_n[\Stream]$ is a bag, the identifiers $I(D_n[\Stream])$ coincide with the positions of the sequence~$t_0 \ldots t_n$.

\changed{Let $Q$ be a CQ} over $\sigma$, and let $\eta: I(Q) \rightarrow I(D_n[\Stream])$ be a t-homomorphism from $Q$ to $D_n[\Stream]$. If we consider $\Omega = I(Q)$, we can interpret $\eta$ as a \emph{valuation} $\hat{\eta}: \Omega \rightarrow 2^\nat$ that maps each atom of $Q$ to a set with a single position; formally, $\hat{\eta}(i) = \{\eta(i)\}$ for every $i \in I(Q)$. Then, we define the semantics of $Q$ over stream $\Stream$ at position $n$ as:
\[
\sem{Q}_n(\Stream) \ = \ \{\hat{\eta} \mid \text{$\eta$ is a t-homomorphism from $Q$ to $D_n[\Stream]$}\}
\]
Note that $\sem{Q}_n(\Stream)$ is equivalent to evaluating $Q$ over $D_n[\Stream]$ where instead of outputting a bag of tuples $\sem{Q}(D_n[\Stream])$, we output the t-homomorphisms (i.e., as valuations) that are in a one-to-one correspondence with the tuples in  $\sem{Q}(D_n[\Stream])$. 

\paragraph{Hierarchical conjunctive queries and main results} 
Let $Q$ be a CQ of the form (\ref{eq:cq}).
Given a variable $x \in \Var$, define $\atoms{x}$ as the bag of all atoms $R_i(\bar{x}_i)$ of $Q$ such that $x$ appears in $\bar{x}_i$. 
We say that $Q$ is \emph{full} if every variable appearing in $\bar{x}_0, \ldots, \bar{x}_{m-1}$ also appears in  $\bar{x}$.
Then, $Q$ is a \emph{Hierarchical Conjunctive Query} (HCQ)\cite{DalviS07a} iff $Q$ is full and for every pair of variables $x, y \in \Var$ it holds that $\atoms{x} \subseteq \atoms{y}$, $\atoms{y} \subseteq \atoms{x}$ or $\atoms{x} \cap \atoms{y} = \emptyset$.
For example, one can check that $\QEXZERO$ is a HCQ, but $\QEXONE$ is not.

HCQ is a subset of CQ that can be evaluated with constant-delay enumeration under updates~\cite{CQUpdates,IdrisUV17}. Moreover, it is the \changed{greatest} class of full conjunctive queries that can be evaluated with such guarantees under fine-grained complexity assumptions. Therefore, HCQ is the right yardstick to measure the expressive power of \acrocea for defining queries with strong efficiency guarantees. 
Given a \acrocea $\ACH$ and a CQ $Q$ over the same schema $\sigma$, we say that $\ACH$ is \emph{equivalent} to $Q$ (denoted as $\ACH \equiv Q$) iff for every stream $\Stream$ over $\sigma$ and every position $n$ it holds that $\sem{\ACH}_n(\Stream) = \sem{Q}_n(\Stream)$.

\begin{theorem}\label{theo:hierarchical-if}
	Let $\sigma$ be a schema. For every HCQ $Q$ over $\sigma$, there exists a \acrocea $\ACH_Q$ over $\sigma$ with unary predicates in $\uncq$ and binary predicates in $\bincq$ such that $\ACH_Q \equiv Q$. Furthermore, $\ACH_Q$ is unambiguous and \changed{of at most} exponential size with respect to $Q$. If $Q$ does not have self joins, then $\ACH_Q$ is of \changed{quadratic}~size. %
\end{theorem}
\changed{
\begin{figure}[t]
	\centering
	\begin{tikzpicture}[>=stealth, 
		semithick, 
		auto,
		initial text= {},
		initial distance= {4mm},
		accepting distance= {3mm},
		mystate/.style={state, inner sep=0pt, minimum size=5mm}]
		
		\node at (-1, 0) {$Q_0(x,y) \leftarrow \underbrace{T(x)}_0, \underbrace{S(x,y)}_1, \underbrace{R(x,y)}_2$};
		
		\path[level distance=1cm,
		level 2/.style={sibling distance=0.8cm}]
		node (raw) at (-1,-1) {$x$} [sibling distance=1.5cm] 
		child { node {$0$} }
		child { node {$y$}
			child { node {$1$} }
			child { node {$2$} }
		};

		\begin{scope}[xshift=3.2cm]

			\node at (-1.3, -1.2) {$\ACH_{Q_0}\!:$};

			\node (ET) at (-0.8, 0) {};
			\node [mystate] (T) at (0, 0) {$0$};
			
			\node (ES) at (-0.8, -1) {};
			\node [mystate] (S) at (0, -1) {$1$};
			
			\node [mystate, accepting] (R) at (2.5, -0.5) {$x$};
			
			\draw[->] (ET) edge node {$T  /  0$} (T);
			\draw[->] (ES) edge node {$S / 1$} (S);

			\draw[->] (T) edge node[sloped] {$(Tx,Rxy)$} (R);
			\draw[->] (S) edge node[sloped, swap] {$(Sx,Rxy)$} (R);
			
			\pic [draw, -, angle radius=1.8cm] {angle = T--R--S};
			\node at ($(R)+(-1.3,0)$) {$R / 2$};
			
		\end{scope}	
		
		\begin{scope}[xshift=7.2cm]

			\node (ET) at (-0.8, 0) {};
			\node [mystate] (T) at (0, 0) {$0$};
			
			\node (ES) at (-0.8, -1) {};
			\node [mystate] (S) at (0, -1) {$2$};
			
			\node [mystate, accepting] (R) at (2.5, -0.5) {$x$};
			
			\draw[->] (ET) edge node {$T / 0$} (T);
			\draw[->] (ES) edge node {$R / 2$} (S);

			\draw[->] (T) edge node[sloped] {$(Tx,Sxy)$} (R);
			\draw[->] (S) edge node[sloped, swap] {$(Rxy,Sxy)$} (R);
			
			\pic [draw, -, angle radius=1.8cm] {angle = T--R--S};
			\node at ($(R)+(-1.3,0)$) {$S /1$};
			
		\end{scope}	
		
		\begin{scope}[xshift=3.2cm, yshift=-2cm]

			\node (ET) at (-0.8, 0) {};
			\node [mystate] (T) at (0, 0) {$1$};
			
			\node (ES) at (-0.8, -1) {};
			\node [mystate] (S) at (0, -1) {$2$};
			
			\node [mystate] (R) at (3.25, -0.5) {$y$};
			
			\node [mystate, accepting] (x) at (6.5, -0.5) {$x$};
			
			\draw[->] (ET) edge node {$S / 1$} (T);
			\draw[->] (ES) edge node {$R / 2$} (S);

			\draw[->] (T) edge node[sloped] {$R,(Sxy,Rxy)/2$} (R);
			\draw[->] (S) edge node[sloped, swap] {$S,(Rxy,Sxy)/1$} (R);
			
			\draw[->] (R) edge node[sloped] {$T,(? xy,Tx)/0$} (x);

		\end{scope}	
	\end{tikzpicture}	
	\caption{An illustration of constructing an \acrocea from a HCQ. On the left, the HCQ $Q_0$ and its $q$-tree. On the right, a \acrocea $\ACH_{Q_0}$ equivalent to $Q_0$. For presentation purposes, states are repeated several times and $?xy$ means a binary relation with any relation name (i.e., $R$ or $S$).}
	\label{fig:ex-proof}
\end{figure}
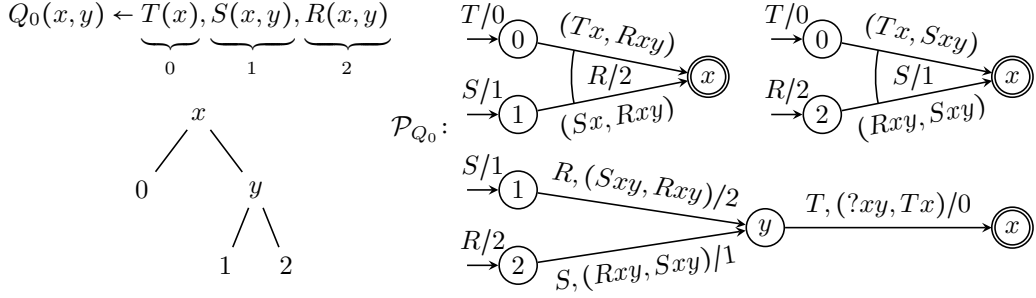
\begin{proof}[Proof sketch]
	We give an example of the construction to provide insights on the expressive power of \acrocea for defining HCQ (the full technical proof is in the appendix). For this construction, we rely on a $q$-tree of a HCQ, a structure introduced in~\cite{CQUpdates}. Formally, let $Q$ be a HCQ and assume, for the sake of simplification, that $Q$ is connected (i.e., the Gaifman graph associated to $Q$ is connected). A \emph{\qt} for $Q$ is a labeled tree, $\qtree: t \rightarrow I(Q) \cup \{\bar{x}\}$, where for every $x \in \{\bar{x}\}$ there is a unique inner node $\bar{u} \in t$ such that $\qtree(\bar{u}) = x$, and for every atom $i \in I(Q)$ there is a unique leaf node $\bar{v} \in t$ such that $\qtree(\bar{v}) = i$. Further, if $\bar{u}_1, \ldots, \bar{u}_k$ are the inner nodes of the path from the root until $\bar{v}$, then $\{\bar{x}_i\} = \{\qtree(\bar{u}_1), \ldots, \qtree(\bar{u}_k)\}$. %
	In~\cite{CQUpdates}, it was shown that a CQ $Q$ is hierarchical and connected iff there exists a $q$-tree for $Q$. For instance, in Figure~\ref{fig:ex-proof} (left) we display again the HCQ $Q_0$, labeled with the identifiers of the atoms, and below a $q$-tree for~$Q_0$. 
	
	For a connected HCQ without self-joins the idea of the construction is to use the $q$-tree of $Q$ as the underlying structure of the \acrocea $\ACH_Q$. Indeed, the nodes of the $q$-tree will be the states of $\ACH_Q$. For example, in Figure~\ref{fig:ex-proof} (right) we present a \acrocea $\ACH_{Q_0}$ equivalent to $Q_0$, where we use multiple copies of the states for presentation purposes (i.e., if two states have the same label, they are the same state in the figure). As you can check, the states are $\{0,1,2,x,y\}$, which are the nodes of the $q$-tree. Furthermore, the leaves of the $q$-trees (i.e., the atoms) are the initial states  $\{0,1,2\}$ where $\ACH_{Q_0}$ uses a unary predicate to check that the tuples have arrived and annotates with the corresponding identifier. 
	
	For every atom $R_i(\bar{x}_i)$ and every variable $x \in \{\bar{x}_i\}$, $\ACH_Q$ jumps with a transition to the state $x$ which is a node in the $q$-tree and joins with all the atoms and variables ``hanging'' from the path from $x$ to the leave $i$ in the $q$-tree. For example, consider the first component (i.e., top-left) of $\ACH_{Q_0}$ in Figure~\ref{fig:ex-proof}. When $\ACH_{Q_0}$ reads a tuple $R(a,b)$, it jumps to state $x$ and joins with all the atoms hanging from the path from $x$ to $2$, namely, the atoms $T$ and $S$. Similarly, consider the last component (i.e., below) of $\ACH_{Q_0}$ in Figure~\ref{fig:ex-proof}. When $\ACH_{Q_0}$ reads a tuple $R(a,b)$, it also jumps to state $y$, but now the only atom hanging from the path from $y$ to $2$ in the $q$-tree is $1$, which corresponds to a single transition from $1$ to $y$ joining with the atom $S(x,y)$. Finally, when $\ACH_{Q_0}$ reads a tuple $T(a)$, the only variable that hangs in the path from the root to $0$ is the variable $y$, and then there is a single transition from $y$ to $x$, joining with an equality predicate $(?xy,Tx)$ where $?xy$ means a binary relation with any relational name (i.e., $R$ or $S$). Finally, the root of the $q$-tree serves as the final state of the $\ACH_{Q_0}$, namely, all atoms were found. Note that an accepting run tree of $\ACH_{Q_0}$ serves as a witness that the $q$-tree is complete. The construction of HCQ with self-joins is more involved, and we present the details in the appendix. 
\end{proof}
}

The previous result shows that \acrocea has the expressive power to specify every HCQ. Given that HCQ characterize the full CQ that can be evaluated in a dynamic setting (under complexity assumptions), a natural question is to ask whether \acrocea has the \emph{right} expressive power, in the sense that it cannot define non-hierarchical CQ. We answer this question positively by focusing on acyclic CQ. Let $Q$ be a CQ of the form~(\ref{eq:cq}). A \emph{join-tree} for $Q$ is labeled tree $\tau:t\rightarrow U(Q)$ such that for every variable $x$ the set $\{\bar{u} \in t \mid \tau(\bar{u}) \in \atoms{x}\}$ form a connected tree in $\tau$. We say that $Q$ is acyclic if $Q$ has a join tree. One can check that both $\QEXZERO$ and $\QEXONE$ are examples of acyclic CQ.

\begin{theorem}\label{theo:hierarchical-onlyif}
	Let $\sigma$ be a schema. For every acyclic CQ $Q$ over $\sigma$, if $Q$ is not hierarchical, then $\ACH \not \equiv Q$ for all \acrocea $\ACH$ over $\sigma$.
\end{theorem}
We note that, although \acrocea can only define acyclic CQ that are hierarchical, it can define queries that are not CQ. For instance, $\PCEAEX$ in Example~\ref{ex:pcea} cannot be defined by any CQ, \changed{since a CQ cannot express that the $R$-tuple must arrive after $T$ and $S$.} Therefore, the class of queries defined by \acrocea~ \changed{is} strictly more expressive than HCQ.

\changed{By Theorems~\ref{theo:hierarchical-if} and~\ref{theo:hierarchical-onlyif}, \acrocea capture the expressibility of HCQ among acyclic CQ}. In the next section, we show that they also \changed{share} their good algorithmic properties for streaming evaluation. 
 	
	\section{An evaluation algorithm for PCEA}\label{sec:algorithm}

Below, we present our evaluation algorithm for unambiguous \acrocea with equality predicates. We do this in a streaming setting where the algorithm reads a stream sequentially, and at each position, we can enumerate the new outputs fired by the last tuple. Furthermore, our algorithm works under a \emph{sliding window} scenario, where we only want to enumerate the outputs inside the last $w$ items for some window size $w$. This scenario is motivated by CER~\cite{GiatrakosAADG20,cugola2012processing,BucchiGQRV22}, where the importance of data decreases with time, and then, we want the outputs inside some relevant time window. 

In the following, we start by defining the evaluation problem and stating the main theorem, followed by describing our data structure for storing valuations. We end this section by explaining the algorithm and stating its correctness. 

\paragraph{The streaming evaluation problem} Let $\Schema$ be a fixed schema. For a valuation $\nu: \Omega \rightarrow 2^{\nat}$, we define $\min(\nu) = \min\{i \mid \exists \ell \in \Omega. \ i \in \nu(\ell)\}$, namely, the minimum position appearing in $\nu$. In this section, we study the following evaluation problem of \acrocea over streams: 
\begin{center}
	\framebox{
		\begin{tabular}{rl}
			\textbf{Problem:} & $\evalprob[\sigma]$\\
			\textbf{Input:} & An unambiguous \acrocea $\ACH = (Q, \uncq, \bincq, \Omega, \Delta, F)$ over $\sigma$,\\
			&  a window size $\window \in \nat$, and a stream $\Stream = t_0 t_1 \ldots $ \\
			\textbf{Output:} & At each position $\ipos$,  enumerate all valuations $\nu \in \sem{\ACH}_\ipos(\Stream)$ \\
			& such that $|\ipos - \min(\nu)| \leq \window$. 
		\end{tabular}
	}
\end{center}
The goal is to output the set $\sem{\ACH}_\ipos^w(\Stream) = \{\nu \in \sem{\ACH}_\ipos(\Stream) \mid |\ipos - \min(\nu)| \leq \window\}$ by reading the stream $\Stream$ tuple-by-tuple sequentially. We assume here a method $\yield{\Stream}$ such that each call retrieves the next tuple, that is, the $\ipos$-th call to $\yield{\Stream}$ retrieves $t_{\ipos}$ for each~$\ipos \geq 0$. 

For solving $\evalprob[\sigma]$, our desire is to find a streaming evaluation algorithm~\cite{grez-chain,IdrisUV17} that, for each tuple $t_\ipos$, \changed{updates} its internal state quickly and enumerates the set $\sem{\ACH}_\ipos^\window(\Stream)$ with output-linear delay. More precisely, let $f: \nat^3 \rightarrow \nat$. A \emph{streaming enumeration algorithm} $\salgo$ with $f$-update time for $\evalprob[\sigma]$ works as follows. Before reading the stream $\Stream$, $\salgo$ receives as input a \acrocea $\ACH$ and $\window \in \nat$, and does some preprocessing. By calling $\yield{\Stream}$, $\salgo$ reads $\Stream$ sequentially and processes the next tuple $t_\ipos$ in two phases called the \emph{update phase} and \emph{enumeration phase}, respectively.
In the update phase, $\salgo$ updates a data structure $\DS$ with $t_\ipos$ taking time $\cO(f(|\ACH|, |t_i|, \window))$.
In the enumeration phase, $\salgo$ uses $\DS$ for enumerating $\sem{\ACH}_\ipos^\window(\Stream)$ with \emph{output-linear delay}. Formally, if $\sem{\ACH}_\ipos^\window(\Stream) = \{\nu_1, \ldots, \nu_{k}\}$ (i.e., in arbitrary order), the algorithm prints $\# \nu_1 \# \nu_2 \# \ldots \# \nu_{k} \#$ to the output registers, sequentially. Furthermore, $\salgo$ prints the first and last symbols $\#$ when the enumeration phase starts and ends, respectively, and the time difference (i.e., the \emph{delay}) between printing the $\#$-symbols surrounding $\nu_i$ is in~$\cO(|\nu_i|)$. Finally, if such an algorithm exists, we say that $\evalprob[\sigma]$ admits a \emph{streaming evaluation algorithm} with $f$-update time and output-linear delay.

In the following, we prove the following algorithmic result for evaluating \acrocea.
\begin{theorem}\label{theo:algorithm}
	$\evalprob[\sigma]$ admits a streaming evaluation algorithm with $(|\ACH|\cdot|t|+|\ACH|\cdot\log(|\ACH|)+|\ACH|\cdot\log(w))$-update time and output-linear delay. 
\end{theorem}
Note that the update time does not depend on the number of outputs seen so far, and regarding data complexity (i.e., assuming that $\ACH$ \changed{and the size of the tuples, $|t|$ are fixed}), the update time is logarithmic in the size of the sliding window.
Theorem~\ref{theo:algorithm} improves with respect to~\cite{grez-chain} by considering a more general class of queries and evaluating over a sliding window. \changed{In contrast}, Theorem~\ref{theo:algorithm} is incomparable to the algorithms for dynamic query evaluation of HCQ in~\cite{CQUpdates,IdrisUV17}. On the one hand, \cite{CQUpdates,IdrisUV17} show
constant update time algorithms for HCQ under insertions and deletions. On the other hand, Theorem~\ref{theo:algorithm} works for CER queries that can compare the order over tuples. If we restrict to HCQ, the algorithms in~\cite{CQUpdates,IdrisUV17} have better complexity, given that there is no need to maintain and check the order of how tuples are inserted or deleted.

\changed{It is important to note that we base the algorithm of Theorem~\ref{theo:algorithm} on the ideas introduced in~\cite{grez-chain}. Nevertheless, it has several new insights that are novel and are not present in~\cite{grez-chain}. First, our algorithm evaluates \acrocea, which is a generalization of CCEA, and then the approach in~\cite{grez-chain} requires several changes. Second, the data structure for our algorithm must manage the evaluation of a sliding window and simultaneously combine parallel runs into one. This challenge requires a new strategy for enumeration that combines cross-products with checking a time condition. Finally, maintaining the runs that are valid inside the sliding window with logarithmic update time requires the design of a new data structure based on the principles of a heap, which is novel. We believe this data structure is interesting in its own right, which could lead to new advances in streaming evaluation algorithms with enumeration.}

We dedicate the rest of this section to explaining the streaming evaluation algorithm of Theorem~\ref{theo:algorithm}, starting by describing the data structure $\DS$.

\paragraph{The data structure} Fix a set of labels $\Omega$. For representing sets of valuations $\nu: \Omega \rightarrow 2^\nat$, we use a data structure composed of nodes, where each node stores a position, a set of labels, and pointers to other nodes. Formally, the \emph{data structure} $\DS$ is composed by a set of nodes, denoted by $\dsnodes$, where each node $\n$ has a set $\dslabels(\n) \subseteq \Omega$, a position $\dspos(\n) \in \nat$, a set $\dsprod(\n) \subseteq \dsnodes$, and two links to other nodes $\dsleft(\n), \dsright(\n) \in \dsnodes$. We assume that the directed graph $G_{\DS}$ with $V(G_{\DS}) = \dsnodes$ and $E(G_{\DS}) = \{(\n_1,\n_2) \mid \n_2 \in \dsprod(\n_1) \vee \n_2 = \dsleft(\n_1) \vee \n_2 = \dsright(\n_1)\}$ is acyclic. In addition, we assume a special node $\dsnull \in \dsnodes$ that serves as a bottom node (i.e., all components above are undefined for $\dsnull$) and $\bot \notin \dsprod(\n)$ for every $\n$. 

Each node in $\DS$ represents a bag of valuations. For explaining this representation, we need to first introduce some algebraic operations on valuations. Given two valuations $\nu, \nu': \Omega \rightarrow 2^\nat$, we define the product $\nu \valop \nu': \Omega \rightarrow 2^\nat$ such that $[\nu \valop \nu'](\ell) = \nu(\ell) \cup \nu'(\ell)$ for every $\ell \in \Omega$. Further, we extend this product to bags of valuations $V$ and $V'$ such that $V \valop V' = \bleft\nu \valop \nu' \mid \nu \in V, \nu' \in V' \bright$.
Note that $\valop$ is an associative and commutative operation and, thus, we can write $\bigvalop_i V_i$ for referring to a sequence of $\valop$-operations.
Given a pair $(L, i) \in 2^\Omega \times \nat$, we define the valuation $\nu_{L, i}: \Omega \rightarrow 2^\nat$ such that
\changed{$\nu_{L,i}(\ell) = \{i\}$ if $\ell \in L$, and $\nu_{L,i} (\ell) = \emptyset$, otherwise}. With this notation, for every $\n \in \dsnodes$ we define the bags $\dssemprod{\n}$ and $\dssem{\n}$ recursively as follows:
\[
\dssemprod{\n} := \bleft\nu_{L(\n), i(\n)}\bright  \valop \bigvalop_{\n' \in \dsprod(\n)} \dssem{\n'}  \ \ \ \ \ \ \ \ \ \ 
\dssem{\n} := \dssemprod{\n} \cup \dssem{\dsleft(\n)} \cup \dssem{\dsright(\n)}.
\]
For $\bot$, we define $\dssemprod{\dsnull} = \dssem{\dsnull} = \emptyset$. Intuitively, the set $\dsprod(\n)$ represents the \emph{product} of its nodes with the valuation $\nu_{L,i}$, and the nodes $\dsleft(\n)$ and $\dsright(\n)$ represent \emph{unions} (for union-left and union-right, respectively). This interpretation is analog to the product and union nodes used in previous work of MSO enumeration~\cite{AmarilliBJM17,MunozR22}, but here we encode products and unions in a single~node. %

For efficiently enumerating $\dssem{\n}$, we require that valuations in $\DS$ are represented without overlapping. To formalize this idea, define that the product $\nu \valop \nu'$ is \emph{simple} if \changed{for every $\ell \in \Omega$, $\nu(\ell)$ and $\nu'(\ell)$ are disjoint and $[\nu \valop \nu'](\ell) = \nu(\ell) \cup \nu'(\ell)$}. Accordingly, we extend this notion to bags of valuations: $V \valop V'$ is simple if $\nu \valop \nu'$ is simple for every $\nu \in V$ and $\nu' \in V'$. We say that $\DS$ is simple if $\bleft\nu_{L(\n), i(\n)}\bright  \valop \bigvalop_{\n' \in \dsprod(\n)} \dssem{\n'}$ is simple for every $\n \in \dsnodes$. This notion is directly related to unambiguous \acrocea in Section~\ref{sec:pcea}. Intuitively, the first condition of unambiguous \acrocea will help us to force that $\DS$ is always simple.

The next step is to incorporate the window-size restriction to $\DS$. For a node $\n \in \dsnodes$, let $\max(\n) = \max\{i \in \nu(\ell) \mid  \nu \in \sem{\n} \wedge \ell \in \Omega \}$. Then, given a position $\ipos \geq \max(\n)$ and a window size $w \in \nat$, define the bag:
\[
\dssem{\n}^w_\ipos \ := \ \bleft \nu \in \dssem{\n} \mid |\ipos - \min(\nu)| \leq \window  \bright.
\]
Our plan is to represent $\dssem{\n}^w_\ipos$ and enumerate its valuations with output-linear delay. For this goal, from now on we fix a $w \in \nat$ and write $\DSw$ to denote the data structure with window size $w$. For the enumeration of $\dssem{\n}^w_\ipos$, in each node $\n$ we store the value:
\[
\dsmaxstart(\n) \ := \ \max\left\{\min(\nu) \mid \nu \in \dssemprod{\n}  \right\}
\]
This value will be helpful to verify whether $\dssem{\n}^w_\ipos$ is non-empty or not; in particular, one can check that $\dssem{\n}^w_\ipos \neq \emptyset$ iff $|\ipos - \dsmaxstart(\n)| \leq \window$. We always assume that $|\max(\n) - \dsmaxstart(\n)| \leq \window$ (otherwise $\dssem{\n}^w_\ipos = \emptyset$).
In addition, we require an order with $\dsleft(\n)$ and $\dsright(\n)$ to discard empty unions easily. For every node $\n \in \dsnodesw$, we require:
\[
\text{$\dsmaxstart(\n) \geq \dsmaxstart(\dsleft(\n))$ \ and \  $\dsmaxstart(\n) \geq \dsmaxstart(\dsright(\n))$} \tag{\ddag} \label{eq:maxstart}
\]
whenever $\dsleft(\n) \neq \bot \neq \dsright(\n)$. Intuitively, the binary tree formed by $\n$ and all nodes that can be reached by following $\dsleft(\cdot)$ and $\dsright(\cdot)$ is not \changed{strictly ordered}; however, it follows the same principle~(\ref{eq:maxstart}) as a \emph{heap}~\cite{cormen2022introduction}. Note that it is not our goal to use $\DSw$ as a priority queue (\changed{since }removing the max element from a heap takes logarithmic time, and we need constant time), but to use condition~(\ref{eq:maxstart}) to quickly check if there are more outputs to enumerate in $\dsleft(\n)$ or $\dsright(\n)$ \changed{by comparing the max-start value of a node with the start of the current location of the time~window.}
\begin{theorem}\label{theo:enumeration}
	Let $w \in \nat$ be a window size and assume that $\DSw$ is simple. Then, for every $\n \in \dsnodesw$ and every position $\ipos \geq \max(\n)$, the valuations in $\dssem{\n}_\ipos^w$ can be enumerated with output-linear delay and without preprocessing (i.e., the enumeration starts immediately). 
\end{theorem}

We require two procedures, called \emph{extend} and \emph{union}, for operating nodes in our algorithm. The first procedure $\dsextend(L,i,\nset)$ receives as input a set $L \subseteq \Omega$, a position $i \in \nat$, and $\nset \subseteq \dsnodesw$ such that $\dspos(\n) < i$ for every $\n \in \nset$. The procedure outputs a fresh node $\n_e$ such that $\dssem{\n_e}^w_\ipos := \bleft\nu_{L, i} \bright  \valop \bigvalop_{\n \in \nset} \dssem{\n}^w_\ipos$. By the construction of $\DSw$, this operation is straightforward to implement by defining $\dslabels(\n_e) = L$, $\dspos(\n_e) = i$, $\dsprod(\n_e) = \nset$, and $\dsleft(\n_e) = \dsright(\n_e)=\dsnull$. Further, we can compute $\dsmaxstart(\n_e)$ from the set $\nset$ as follows: $\dsmaxstart(\n_e) = \min\{i, \min\{ \dsmaxstart(\n) \mid \n \in \nset\}\}$. Overall, we can implement $\dsextend(L,i,\nset)$ with running time $\cO(|\nset|)$. 

The second procedure $\dsunion(\n_1, \n_2)$ receives as inputs two nodes $\n_1, \n_2 \in \dsnodesw$ such that $\max(\n_1) \leq \dspos(\n_2)$ and $\dsleft(\n_2) = \dsright(\n_2) = \bot$. It outputs a fresh node $\n_u$ such that $\dssem{\n_u}^w_\ipos := \dssem{\n_1}^w_\ipos \cup \dssem{\n_2}^w_\ipos$. The implementation of this procedure is more involved since it requires inserting $\n_2$ into $\n_1$ by using $\dsleft(\n_1)$ and $\dsright(\n_1)$, and maintaining condition~(\ref{eq:maxstart}). Furthermore, we require them to be fully persistent~\cite{driscoll1986making}, namely, $\n_1$ and $\n_2$ are unmodified after each operation.  
\begin{proposition}\label{prop:union}
	Let $k \in \nat$ and assume that one performs $\dsunion(\n_1, \n_2)$ over $\DSw$ with the same position $i = \dspos(\n_2)$ at most $k$ times. Then one can implement $\dsunion(\n_1, \n_2)$ with running time $\cO(\log(k\cdot w))$ per call. 
\end{proposition}

\begin{algorithm}[t]
	\caption{Evaluation of an unambiguous \acrocea $\ACH = (Q, \uncq, \bincq, \Omega, \Delta, F)$ with equality predicates over a stream $\Stream$ under a sliding window of size $w$.}\label{alg:parallel-eval}
		\smallskip
		\begin{varwidth}[t]{0.5\textwidth}
			\begin{algorithmic}[1]
				\Procedure{{Evaluation}}{$\ACH, w, \Stream$}
				\State $\DSw \gets \emptyset$
				\State $\ipos \gets -1$
				\While{$t \gets \yield{\Stream}$}
				\State \Call{Reset}{$\blank$}
				\State \Call{FireTransitions}{$t, i$}
				\State \Call{UpdateIndices}{$t, \ipos$}
				\ForEach{$\n \in \bigcup_{p\in F} \nsetq{p}$}
				\State \Call{Enumerate}{$\n, \ipos, w$}
				\EndFor
				\EndWhile
				\EndProcedure
				\State
				\Procedure{Reset}{\blank}
				\State $\ipos \gets i+1$
				\ForEach{$p \in Q$}
				\State $\nsetq{p} \gets \emptyset$
				\EndFor
				\EndProcedure
				\algstore{myalg}
			\end{algorithmic}
		\end{varwidth} \hfill
		\begin{varwidth}[t]{0.6\textwidth}
			\begin{algorithmic}[1]
				\algrestore{myalg}
				\Procedure{{FireTransitions}}{$t, \ipos$}
				\ForEach{$e = (P, U, \binfunc, L, q) \in \Delta$}
				\If{$t \in U \wedge \bigwedge_{p \in P} \htable{e}{p}{\rightbinfunc{\binfunc}_p(t)} \neq \emptyset$}
				\State $\nset \gets \{\, \htable{e}{p}{\rightbinfunc{\binfunc}_p(t)} \mid p \in P \,\}$
				\State $\nsetq{q} \gets \nsetq{q} \cup \{\dsextend(L,\ipos,\nset)\}$
				\EndIf
				\EndFor
				\EndProcedure
				\State
				\Procedure{{UpdateIndices}}{$t$}
				\ForEach{$ e = (P, U, \binfunc, L, q) \in \Delta$}
				\ForEach{$p \in P \wedge \n \in \nsetq{p}$}
				\If{$\htable{e}{p}{\leftbinfunc{\binfunc}_p(t)} = \emptyset$}
				\State $\htable{e}{p}{\leftbinfunc{\binfunc}_p(t)} \gets \n$
				\Else
				\State $\n' \gets \htable{e}{p}{\leftbinfunc{\binfunc}_p(t)}$
				\State $\htable{e}{p}{\leftbinfunc{\binfunc}_p(t)} \gets \dsunion(\n', \n)$
				\EndIf 
				\EndFor
				\EndFor
				\EndProcedure
			\end{algorithmic}
		\end{varwidth} 
	\end{algorithm}

\paragraph{The streaming evaluation algorithm} In Algorithm~\ref{alg:parallel-eval}, we present the main procedures of the evaluation algorithm given a fixed schema $\sigma$. The algorithm receives as input a \acrocea $\ACH = (Q, \uncq, \bincq, \Omega, \Delta, F)$ over $\sigma$, a window size $w \in \bbN$, and a reference to a stream $\Stream$. We assume that these inputs are globally accessible by all procedures. Recall that we can test if $t \in U$ in linear time for any $U \in \uncq$. Further, recall that $\bincq$ are equality predicates and, for every $B \in \bincq$, there exists linear time computable partial functions $\leftbinfunc{B}$ and $\rightbinfunc{B}$ such that $(t_1, t_2) \in B$ iff $\leftbinfunc{B}(t_1)$ and $\rightbinfunc{B}(t_2)$ are defined and $\leftbinfunc{B}(t_1) = \rightbinfunc{B}(t_2)$, for every $t_1,t_2 \in \tuples[\Schema]$.

For the algorithm, we require some data structures. First, we use the previously described data structure $\DSw$ and its nodes $\dsnodesw$. Second, we consider a \emph{look-up table} $\htablesym$ that maps triples of the form $(e, p, d)$ to nodes in $\dsnodesw$ where $e \in \Delta$, $p\in Q$, and $d$ is the output of any partial function $\leftbinfunc{B}$ or $\rightbinfunc{B}$. We write $\htable{e}{p}{d}$ for accessing its node, and $\htable{e}{p}{d} \gets \n$ for updating a node $\n$ at entry $(e,p,d)$. Also, we write $\htable{e}{p}{d} = \emptyset$ or  $\htable{e}{p}{d} \neq \emptyset$ for checking whether there is a node or not at entry $(e, p,d)$. We assume all entries are empty at the beginning. Intuitively, for $e = (P, U, \binfunc, L, q) \in \Delta$ and $p \in P$, we use $\htable{e}{p}{\cdot}$ to check if the equality predicate $\binfunc_p$ is satisfied or not (here $\binfunc_p = \binfunc(p)$). As it is standard in the literature~\cite{CQUpdates, IdrisUV17} (i.e., by adopting the RAM model), we assume that each operation over look-up tables takes constant time. Finally, we assume a set of nodes $\nsetq{p}$ for each $p \in Q$ whose use will be clear later. 

Algorithm~\ref{alg:parallel-eval} starts at the main procedure $\textsc{Evaluation}$. It initializes the data structure $\DSw$ to empty (i.e., the only node it has is the special node $\dsnull$) and the index $\ipos$ for keeping the current position in the stream (lines 2-3). Then, the algorithm loops by reading the next tuple $\yield{\Stream}$, performs the update phase (lines 5-7), followed by the enumeration phase (lines 8-9), and repeats the process over again. Next, we explain the update phase and enumeration phase separately.

The update phase is composed of three steps, encoded as procedures. The first one, $\textsc{Reset}$, is in charge of starting a new iteration by updating $\ipos$ to the next position and emptying the sets $\nsetq{p}$ (lines 12-14). The second step, $\textsc{FireTransitions}$, uses the new tuple $t$ to fire all transitions $e = (P, U, \binfunc, L, q) \in \Delta$ of $\ACH$ (lines 16-19). We do this by checking if $t$ satisfies $U$ and all equality predicates $\{\binfunc_p\}_{p\in P}$ (line 17). The main intuition is that the algorithm stores partial runs in the look-up table $\htablesym$, whose outputs are represented by nodes in $\DSw$. Then the call $\htable{e}{p}{\rightbinfunc{\binfunc}_p(t)}$ is used to verify the equality $\leftbinfunc{\binfunc}_p(t') = \rightbinfunc{\binfunc}_p(t)$ for some previous tuple $t'$. Furthermore, if $\htable{e}{p}{\rightbinfunc{\binfunc}_p(t)}$ is non-empty, it contains the node that represents all runs that have reached $p$. If $U$ and all predicates $\{\binfunc_p\}_{p\in P}$  are satisfied, we collect all nodes at states $P$ in the set $\nset$ (line 18), and symbolically extend these runs by using the method $\dsextend(L,\ipos,\nset)$ of $\DSw$. We collect the output node of $\dsextend$ in the set $\nsetq{q}$ for use in the next procedure \textsc{UpdateIndices}. 

The last step of the update phase, \textsc{UpdateIndices}, is to update the look-up table $\htablesym$ by using $t$ and the nodes stored at the sets $\{\nsetq{p}\}_{p \in Q}$ (lines 22-28). Intuitively, the nodes in $\nsetq{p}$ represent new runs (i.e., valuations) that reached state $p$ when reading $t$. Then, for every transition $e = (P, U, \binfunc, L, q) \in \Delta$ such that $p \in P$, we want to update the entry $(e,p,\leftbinfunc{\binfunc}_p(t))$ of $\htablesym$ with the nodes from $\nsetq{p}$, to be ready to be fired for future tuples. 
For this goal, we check each $\n \in \nsetq{p}$ and, if $\htable{e}{p}{\leftbinfunc{\binfunc}_p(t)}$ is empty, we just place $\n$ at the entry $(e,p,\leftbinfunc{\binfunc}_p(t))$ (lines~23-25). Otherwise, we use the union operator of $\DSw$, to combine the previous outputs with the new ones of $\n$ (lines 26-28). Note that the call to $\dsunion(\n', \n)$ satisfies the requirements of this operator, given that $\n$ was created recently.

Based on the previous description, the enumeration phase is straightforward. Given that the nodes in $\{\nsetq{p}\}_{p \in Q}$ represent new runs at the last position, $\bigcup_{p\in F} \nsetq{p}$ are all new runs that reached some final state. Then, for each node $\n \in \bigcup_{p\in F} \nsetq{p}$ we call the procedure $\textsc{Enumerate}(\n, \ipos, w)$ that enumerates all valuations in $\dssem{\n}^w_i$. %
Theorem~\ref{theo:enumeration} shows that this method exists with the desired guarantees given that $\ACH$ is unambiguous which implies that $\DSw$ is simple. Further, runs correspond with valuations, namely, $\dssem{\n}^w_i$ is a set, and, thus, we enumerate the outputs without repetitions. 

\begin{proposition}\label{prop:algo-correctness}
	For every unambiguous \acrocea $\ACH$ with equality predicates, $w \in \bbN$, stream $\Stream$, and position~$\ipos \in \bbN$, Algorithm~\ref{alg:parallel-eval} enumerates all valuations $\sem{\ACH}_\ipos^w(\Stream)$ without repetitions.
\end{proposition}

We end by discussing the update time of Algorithm~\ref{alg:parallel-eval}. By inspection, one can check that we performed a linear pass over $\Delta$  during the update phase, where each iteration takes linear time over each transition. Overall, we made at most $\cO(|\ACH|)$ calls to unary predicates, the look-up table, or the data structure $\DSw$. Each call to a unary predicate takes $\cO(|t|)$-time and, thus, at most $\cO(|\ACH|\cdot |t|)$-time in total. The operations to the look-up table or $\dsextend$ take constant time. Instead, we performed at most $\cO(|\ACH|)$ unions over the same position $\ipos$. By Proposition~\ref{prop:union}, each $\dsunion$ takes time $\cO(\log(|\ACH|\cdot w))$. Summing up, the updating time %
is $\cO(|\ACH|\cdot|t| + |\ACH|\cdot\log(|\ACH|) + |\ACH|\cdot\log(w))$.

	\section{Future work}\label{sec:conclusions}

We present an automata model for CER that \changed{expresses} HCQ and can be evaluated in a streaming fashion under a sliding window with a logarithmic update and output-linear delay. These results achieve the primary goal of this paper but leave several directions for future work. First, defining a query language that characterizes the expressive power of \acrocea will be interesting. 
Second, one would like to understand a disambiguation procedure to convert any \acrocea into an unambiguous \acrocea or to decide when this is possible. Last, we study here algorithms for \acrocea with equality predicates, but the model works for any binary predicate. Then, it would be interesting to understand for which other predicates (e.g., inequalities) the model still admits efficient streaming evaluation.

	\bibliographystyle{abbrv}
	\bibliography{biblio}
	
	\newpage
	\appendix
	
	\section{Proofs of Section~\ref{sec:pcea}} \label{sec:app-pcea}

\subsection*{Proof of Proposition~\ref{prop:pfa-are-regular}}

\begin{proof}
	To prove this statement, we follow the same principle used in the subset construction. To simulate all possible run trees of a \aacro with a DFA, we start at the leaves, with all initial states. Then for each symbol we move up on the tree, firing all transitions that used a subset of the current set of states. At the end of the string, if the last set has a final state, then it means that one can construct a run tree that accepts the input.

	Let $\cP = (Q, \Sigma, \Delta, I, F)$ be a \anamelc. We build the DFA $\cA = (2^Q, \Sigma, \delta, I, F')$ such that $F' = \{P \mid P \cap F \neq \emptyset\}$ and $\delta(P, a) = \{q \mid \exists P' \subseteq P. \, (P', a, q) \in \Delta)\}$ for every $P \subseteq Q$ and $a \in \Sigma$. We now prove that both automata define the same language.

	\paragraph{$\cL(\cP) \subseteq \cL(\cA)$}
	Let $\bar{s}=a_1\ldots a_n \in \Sigmastar$ be a string such that $\bar{s} \in \cL(\cP)$ and let $\tau: t \rightarrow Q$ be an accepting run tree of $\cP$ over $\bar{s}$. We need to prove that the run $\rho: S_n \xrightarrow{a_1} S_{n-1} \xrightarrow{a_2} \ldots \xrightarrow{a_n} S_0$ is an accepting run of $\cA$ over $\bar{s}$, i.e. $S_n \in F'$.
	To this end, we define $L_i = \{\tau(\bar{u}) \mid \depth{\tau}{\bar{u}} = i\}$ as the set of states labeling $\tau$ at depth $i$ and prove that $L_i \subseteq S_i$ for all $0\leq i \leq n$. Since $L_0 = \{\tau(\varepsilon)\}$, this in return means that $S_n \cap F \neq \emptyset$ and $S_n \in F'$.

	For every leaf node $\bar{u}$ it holds that $\depth{\tau}{\bar{u}}=n$ and $\tau(\bar{u})\in I$, meaning $L_n \subseteq S_n = I$. Let us assume that $L_{i-1} \subseteq S_{i-1}$; for every inner node $\bar{v}$ at depth $i$ there must be a transition $(P, a_{n-i}, q) \in \Delta$ such that $\tau(\bar{v}) = q$ and $P = \{\tau(\bar{u}) \mid \bar{u} \in \children{\tau}{\bar{v}}\}$. Following the definition of $\delta$, it is clear that $q \in \delta(P, a)$, and since this is true for every node at depth $i$, we have that $L_i \subseteq S_i$.

	Given that $L_0 \subseteq S_0$, we know that $S_n \in F'$, which means that $\rho$ is an accepting run of $\cA$ over $\bar{s}$ and therefore $\cL(\cP) \subseteq \cL(\cA)$.

	\paragraph{$\cL(\cA) \subseteq \cL(\cP)$}
	Let $\bar{s}=a_1\ldots a_n \in \Sigmastar$ be a string such that $\bar{s} \in \cL(\cA)$ and let $\rho: S_n \xrightarrow{a_1} S_{n-1} \xrightarrow{a_2} \ldots \xrightarrow{a_n} S_0$ be the run of $\cA$ over $\bar{s}$. We can now construct a run tree of $\cP$ over $\bar{s}$.

	Since $\rho$ is an accepting run, we know that $S_0 \cap F \leq \emptyset$. We define $\tau: t\rightarrow Q$ such that $\tau(\varepsilon) = f$ with $f \in S_0 \cap F$. If we consider a node $\overline{v} \in t$ at depth $i$, such that $\tau(\bar{v}) = q$ and $q\in S_i$, we can follow the definition of $\delta$, and inductively add nodes to $\tau$ according to the transition $(P, a_{n-i}, q) \in \Delta$ so that $|\children{\tau}{\bar{v}}| = |P|$ and $P = \{\tau(\bar{u}) \mid \bar{u} \in \children{\tau}{\bar{v}}\}$. For every leaf node $\bar{v}$ it holds that $\depth{\tau}{\bar{v}} = n$ and since $S_n = I$ all of them will be labeled by initial states.

	The labeled tree $\tau$ we just constructed is an accepting run of $\cP$ over $\bar{s}$, meaning $\cL(\cA) \subseteq \cL(\cP)$ and, therefore, $\cL(\cP) = \cL(\cA)$.
\end{proof}

\subsection*{Proof of Proposition~\ref{prop:pcea-expressive}}
\begin{proof}
	To prove this statement we just need to find a \anamecea $\ACH$ with no CCEA equivalent, i.e. there is no CCEA $\cC$ such that $\sem{\ACH}(\Stream) = \sem{\cC}(\Stream)$ for every stream~$\Stream$.
	\changed{Let $\ACH$ be the \acrocea represented in Figure~\ref{fig:ex-proof}, then }$\ACH = \big(Q, \un, \bin, \Omega, \Delta, F \big)$, with $Q = \{R(x, y), S(x, y), T(x), x, y\}$, $\Omega = \{R, S, T\}$, $F=\{x\}$ and:
	\[
	\begin{array}{rcl}
	\Delta & = & \big\{ (\emptyset, U_{R(x, y)}, \emptyset, \{R(x, y)\}, R(x, y)), \\
		& &(\emptyset, U_{S(x, y)}, \emptyset, \{S(x, y)\}, S(x, y)), \\
		& &(\emptyset, U_{T(x)}, \emptyset, \{T(x)\}, T(x)), \\
		& &(\{R(x, y), T(x)\}, U_{S(x, y)}, \{(R(x, y), B_{R(x, y), S(x, y)}), (T(x), B_{T(x), S(x, y)})\}, \{S(x, y)\}, x), \\
		& &(\{S(x, y), T(x)\}, U_{R(x, y)}, \{(S(x, y), B_{S(x, y), R(x, y)}), (T(x),B_{T(x), R(x, y)})\}, \{R(x, y)\}, x),  \\   
		& &(\{R(x, y)\}, U_{S(x, y)}, \{(R(x, y), B_{R(x, y), S(x, y)})\}, \{S(x, y)\}, y), \\
		& &(\{S(x, y)\}, U_{R(x, y)}, \{(S(x, y), B_{S(x, y), R(x, y)})\}, \{R(x, y)\}, y), \\
		& &(\{y\}, U_{T(x)}, \{(y, B_{y, T(x, y)})\}, \{T(x)\}, x) \big\}
	\end{array}
	\]
	with the predicates $U_{R(\bar{x})}$ and $B_{R(\bar{x}), S(\bar{y})}$ defined as:
	\[
		U_{R(\bar{x})} \ := \  \{R(\bar{a}) \in \tuples[\Schema] \mid \exists h \in \Hom.\; h(R(\bar{x})) = R(\bar{a})\}
	\]
	and:
	\[
		B_{R(\bar{x}), S(\bar{y})} \ := \ \{(R(\bar{a}), S(\bar{b})) \mid \exists h \in \Hom.\; h(R(\bar{x})) = R(\bar{a}) \wedge h(S(\bar{y})) = S(\bar{b})\}.
	\]

	Let $\Stream_i = \bleft R(0, i), T(0), S(0, i), \ldots \bright$ be a family of streams over the set of data values $\Data=\nat$ with $i \in \nat$. It is clear that the valuation $\bleft 0, 1, 2\bright \in \sem{\ACH}(\Stream_i)$ for every $i \in \nat$.
	Let $\cC = (Q', \un', \bin', \Omega', \Delta', I', F')$ be a deterministic CCEA such that $\sem{\CHCEA}(\Stream_i) = \sem{\ACH}(\Stream_i)$ for every $i \in \nat$. This means that for every stream $\Stream_i$, there is an accepting run of $\CHCEA$ over of the form $\rho_i: q_{i, 0} \xrightarrow{R(0, i)} q_{i, 1} \xrightarrow{T(0)} q_{i, 2} \xrightarrow{S(0, i)} q_{i, 3}$.

	Since $\cC$ has a finite number of states, we know that there must be two streams, $\Stream_j$ and $\Stream_k$ with $j \neq k$ with accepting runs $\rho_j: q_{j, 0} \xrightarrow{R(0, j)} q_{j, 1} \xrightarrow{T(0)} q_{j, 2} \xrightarrow{S(0, j)} q_{j, 3}$ and $\rho_k: q_{k, 0} \xrightarrow{R(0, k)} q_{k, 1} \xrightarrow{T(0)} q_{k, 2} \xrightarrow{S(0, k)} q_{k, 3}$, respectively, such that $q_{j, i} = q_{k, i}$ for every $0 \leq i \leq 3$.

	Given the run $\rho_k$ of $\cC$, we know that there must be a transition $(q_{k, 2}, U, B, \omega, q_{k, 3})\in \Delta'$ such that $S(0, k) \in U$ and $(T(0), S(0, k)) \in B$ and since $q_{k, 2} = q_{j, 2}$ and $q_{k, 3} = q_{j, 3}$ the following will be an accepting run of $\cC$ over the stream $S_{j, k} = \{\{ R(0, j), T(0), S(0, k)\}\}$: $\rho_{j, k}: q_{j, 0} \xrightarrow{R(0, j)} q_{j, 1} \xrightarrow{T(0)} q_{j, 2} \xrightarrow{S(0, k)} q_{j, 3}$.

	We can easily check that there are no accepting runs of $\ACH$ over $\Stream_{j, k}$, meaning $\sem{\ACH}(\Stream_{j, k}) \neq \sem{\cC}(\Stream_{j, k})$ and therefore there is no CCEA $\cC$ such that $\sem{\ACH}(\Stream) = \sem{\cC}(\Stream)$ for every stream~$\Stream$.
\end{proof} 	
	\section{Proofs of Section~\ref{sec:hierarchical}} \label{sec:app-hierarchical}
\subsection*{Proof of equivalence between CQ bag-semantics}
	Fix a schema $\sigma$, a set of data values $\Data$ and a CQ $Q$ over $\sigma$ of the form:
	\[
		Q(\bar{x}) \ \leftarrow \ R_0(\bar{x}_0), \ldots, R_{m-1}(\bar{x}_{m-1})
	\]
	To prove that $\sem{Q}(D) = \semaux{Q}(D)$ we need to prove that $\sem{Q}(D) \subseteq \semaux{Q}(D)$ and $\semaux{Q}(D) \subseteq \sem{Q}(D)$, where $\semaux{Q}(D) \subseteq \sem{Q}(D)$ if $U(\semaux{Q}(D)) \subseteq U(\sem{Q}(D))$ and \changed{for every $a$ in $I(\semaux{Q}(D))$, $\mult{\semaux{Q}(D)}{a} \leq \mult{\sem{Q}(D)}{a}$.
	
	Both $\sem{Q}(D)$ and $\semaux{Q}(D)$ map every atom of $Q$ to the database $D$, meaning $U(\semaux{Q}(D)) = U(\sem{Q}(D))$, so now for every $Q(\bar{a})$ in $I(\semaux{Q}(D))$ we just need to prove that for every $a$ in $I(\semaux{Q}(D))$, $\mult{\semaux{Q}(D)}{Q(\bar{a})} = \mult{\sem{Q}(D)}{a}$.}

	By following the definitions given previously for the multiplicities we get:
	\[
	\begin{aligned}
		\mult{\sem{Q}(D)}{Q(\bar{a})}
		& = \big|\{j \mid \sem{Q}(D)(j) = Q(\bar{a})\}\big|\\
	\end{aligned}
	\]
	We also know that if $\sem{Q}(D)(j) = Q(\bar{a})$ holds, there must be a t-homomorphism $\eta$ such that $Q(h_\eta(\bar{x})) = Q(\bar{a})$. Since every t-homomorphism can map more than one atom, it is clear~that:
	\[
	\begin{aligned}
		\mult{\sem{Q}(D)}{Q(\bar{a})}
		& = \bigg\lvert\bigcup_{\substack{\eta \in \text{t-Hom}(Q, D)\,\colon\\ h_\eta(\bar{x}) = \bar{a}}} \{j \mid \sem{Q}(D)(j) = Q(h_\eta(\bar{x}))\} \bigg\rvert \\
		& = \sum_{\substack{\eta \in \text{t-Hom}(Q, D)\,\colon\\ h_\eta(\bar{x}) = \bar{a}}} \big|\{j \mid \sem{Q}(D)(j) = Q(h_\eta(\bar{x}))\}\big| \\
		& = \sum_{\substack{\eta \in \text{t-Hom}(Q, D)\,\colon\\ h_\eta(\bar{x}) = \bar{a}}} \mult{Q,D}{h_\eta} \\
	\end{aligned}
	\]
	As was stated before, there is a correspondence between t-homomorphisms and homomorphism, meaning:
	\[
	\begin{aligned}
		\mult{\semaux{Q}(D)}{Q(\bar{a})}
		& = \sum_{\substack{h \in \text{Hom}(Q, D)\,\colon\\ h(\bar{x}) = \bar{a}}} \mult{Q,D}{h} \\
		& = \mult{\semaux{Q}(D)}{Q(\bar{a})}
	\end{aligned}
	\]
	Given that $\mult{\semaux{Q}(D)}{Q(\bar{a})} = \mult{\sem{Q}(D)}{Q(\bar{a})}$ for every $Q(\bar{a})$ in $I(\semaux{Q}(D))$, \changed{this implies that $\mult{\semaux{Q}(D)}{Q(\bar{a})} = \mult{\sem{Q}(D)}{Q(\bar{a})}$ and therefore, for every CQ $Q$ and database $D$ it holds that $\sem{Q}(D) = \semaux{Q}(D)$.}

\subsection*{Proof of Theorem~\ref{theo:hierarchical-if}}

\begin{proof}
	Fix a schema $\sigma$ and fix a HCQ $Q$ over $\sigma$ of the form:
	\[
	Q(\bar{x}) \ \leftarrow \ R_0(\bar{x}_0), \ldots, R_{m-1}(\bar{x}_{m-1}).
	\]
	Without loss of generality, we assume that $Q$ has at least two atoms, i.e., $m \geq 2$. If not, it is straightforward to construct a \acrocea for $Q$. Further, for the sake of simplification, we assume that $Q$ does not have data values (i.e., constants); all the constructions below work with data values in the atoms with the additional cost of differentiating the variables from the data values in the set $\{\bar{x}_i\}$. Given that $Q$ is full, this means that $\{\bar{x}\} = \bigcup_{i = 0} \{\bar{x}_i\}$. For this reason, we will use $\{\bar{x}\}$ to refer to the set of all variables in $Q$. Recall that we usually consider $Q$ as bag of atoms, namely, $Q = \bleft  R_0(\bar{x}_0), \ldots, R_{m-1}(\bar{x}_{m-1})\bright$. Then $I(Q)$ is the set of all the identifiers $\{0, \ldots, m-1\}$ and $U(Q)$ the set of all different atoms in~$Q$. We say that $Q$ is \emph{connected}\footnote{Given that $Q$ is hierarchical, this definition is equivalent to the notion of connected CQ, i.e., that the Gaifman graph of $Q$ is connected.} iff there exists a variable $x \in \{\bar{x}\}$ such that $x \in \{\bar{x}_i\}$ for every~$i < m$.
	
	In the following, we first give the proof of the theorem for connected HCQ without self joins, then move to the case of connected HCQ with self joins, and end with the general~case.
	
	\paragraph{Connected HCQ without self joins}  A \emph{\qt} for $Q$ is a labeled tree, $\qtree: t \rightarrow I(Q) \cup \{\bar{x}\}$, where for every $x \in \{\bar{x}\}$ there is a unique inner node $\bar{u} \in t$ such that $\qtree(\bar{u}) = x$, and for every $i \in I(Q)$ there is a unique leaf node $\bar{v} \in t$ such that $\qtree(\bar{v}) = i$ and, if $\bar{u}_1, \ldots, \bar{u}_k$ are the inner nodes of the path from the root until $\bar{v}$,  then $\{\bar{x}_i\} = \{\qtree(\bar{u}_1), \ldots, \qtree(\bar{u}_k)\}$, i.e., $\{\bar{x}_i\} = \{\qtree(\bar{u}) \mid \bar{u} \in \ancst{\qtree}{\bar{v}} \wedge \bar{u} \neq \bar{v}\}$. As an example, in Figure \ref{fig-q-tree} we can see a \qt associated with $Q_1(x, y, z, v, w) \leftarrow \bleft R(x, y, z), S(x, y, v), T(x, w), U(x, y) \bright$ and two valid \qt{s} for the query $Q_2(x, y, z, v) \leftarrow \bleft R(x, y, z), R(x, y, v), U(x, y)\bright$. Note that we use the identifiers of the atoms (instead of the atom itself), so if an atom appears $n$ times in the query, there will be $n$ different leaf nodes for the same atom. 
	\begin{figure}[t]
		\centering
\begin{tikzpicture}[-,>=stealth',roundnode/.style={circle,draw,inner sep=1.2pt},squarednode/.style={rectangle,inner sep=3pt}]
	\node [squarednode] (R) at (1, 0) {$R$};
	\node [squarednode] (S) at (2, 0) {$S$};
	\node [squarednode] (U) at (0, 1) {$U$};
	\node [squarednode] (z) at (1, 1) {$z$};
	\node [squarednode] (v) at (2, 1) {$v$};
	\node [squarednode] (T) at (3, 1) {$T$};
	\node [squarednode] (y) at (1, 2) {$y$};
	\node [squarednode] (w) at (3, 2) {$w$};
	\node [squarednode] (x) at (2, 3) {$x$};
	
	\draw (x) to (y);
	\draw (x) to (w);
	\draw (y) to (U);
	\draw (y) to (z);
	\draw (y) to (v);
	\draw (w) to (T);
	\draw (z) to (R);
	\draw (v) to (S);
	
	\node [left=2cm of x] {$\tau_{Q_1}$:};
\end{tikzpicture}
\begin{tikzpicture}[-,>=stealth',roundnode/.style={circle,draw,inner sep=1.2pt},squarednode/.style={rectangle,inner sep=3pt}]
	\node [squarednode] (R) at (1, 0) {$R_{xyz}$};
	\node [squarednode] (R2) at (2, 0) {$R_{xyv}$};
	\node [squarednode] (U) at (0, 1) {$U_{xy}$};
	\node [squarednode] (z) at (1, 1) {$z$};
	\node [squarednode] (v) at (2, 1) {$v$};
	\node [squarednode] (y) at (1, 2) {$y$};
	\node [squarednode] (x) at (1, 3) {$x$};
	
	\draw (x) to (y);
	\draw (y) to (U);
	\draw (y) to (z);
	\draw (y) to (v);
	\draw (z) to (R);
	\draw (v) to (R2);

	\node [left=1cm of x] {$\tau_{Q_2}$:};
\end{tikzpicture}
\begin{tikzpicture}[-,>=stealth',roundnode/.style={circle,draw,inner sep=1.2pt},squarednode/.style={rectangle,inner sep=3pt}]
	\node [squarednode] (R) at (1, 0) {$R_{xyz}$};
	\node [squarednode] (R2) at (2, 0) {$R_{xyv}$};
	\node [squarednode] (U) at (0, 1) {$U_{xy}$};
	\node [squarednode] (z) at (1, 1) {$z$};
	\node [squarednode] (v) at (2, 1) {$v$};
	\node [squarednode] (y) at (1, 2) {$x$};
	\node [squarednode] (x) at (1, 3) {$y$};
	
	\draw (x) to (y);
	\draw (y) to (U);
	\draw (y) to (z);
	\draw (y) to (v);
	\draw (z) to (R);
	\draw (v) to (R2);

	\node [left=1cm of x] {$\tau_{Q_2}'$:};
\end{tikzpicture}

 		\caption{Example of valid \qt{s} for $Q_1$ and $Q_2$. For the sake of readability we represent each node $\bar{u}$ by the variable/atom in the underlying set of $Q$ that labels the node, $Q(\qtree(\bar{u}))$. Since we have no self joins in $Q_1$, we omit the variables when representing the atoms.}
		\label{fig-q-tree}
	\end{figure}
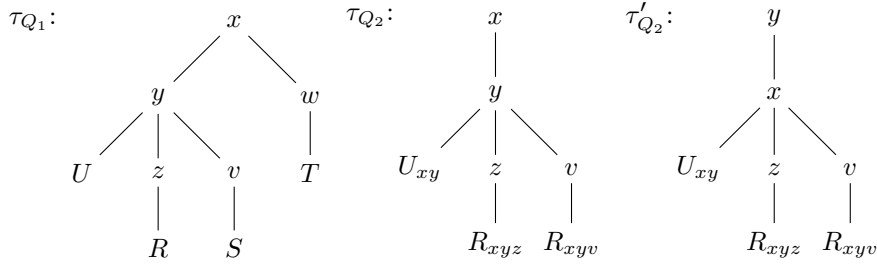
	\begin{theorem}[\cite{CQUpdates}]
		$Q$ is hierarchical and connected iff there exists a \qt for $Q$.
	\end{theorem}
	In the following, we fix a \qt for $Q$ that we denote by $\qtree$. Furthermore, given that there is a one-to-one correspondence between nodes of $\qtree$ and the set $I(Q) \cup \{\bar{x}\}$, by some abuse of notation, we will usually use variables and identifiers as nodes in~$\qtree$~(e.g.,~$\desc{\qtree}{x}$ or~$\ancst{\qtree}{i}$).  
	
	We define the \emph{\cqt} of an HCQ $Q$ as the result of taking the original \qt of the query and removing all inner nodes with a single child. In other words, to generate the \cqt $\qtree^c$ from $\qtree$, we copy $\qtree$, then for each node $\bar{u} \in \qtree^c$ such that $\children{\qtree^c}{\bar{u}} = \{\bar{v}\}$, we remove $\bar{v}$ and \changed{$\qtree^c(\bar{v})$ from $\qtree^c$, and we redefine $\children{\qtree^c}{\bar{u}} = \children{\qtree^c}{\bar{v}}$}. We can see in Figure~\ref{fig-compact-tree} the \cqt{s} associated with each of the \qt{s} of Figure~\ref{fig-q-tree}. Note that for every \qt $\qtree$, the root of  $\cqtree$ is always variable and has at least two children since we assume that $Q$ has at least two atoms (and then $\qtree$ has at least two leaves). For simplification, from now on, we assume that $\qtree$ is compact. If not, all the constructions below hold verbatim by replacing $\qtree$ by $\cqtree$ and restricting to the variables that appear in $\cqtree$.
 	
	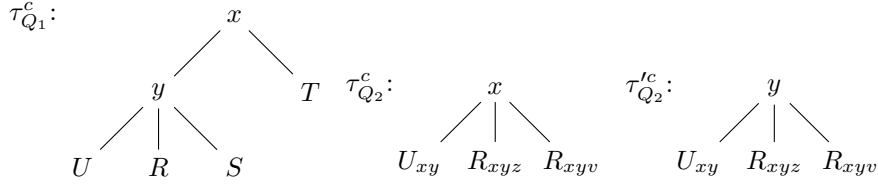
\begin{figure}[t]
		\centering
\begin{tikzpicture}[-,>=stealth',roundnode/.style={circle,draw,inner sep=1.2pt},squarednode/.style={rectangle,inner sep=3pt}]
	\node [squarednode] (U) at (1, 0) {$U$};
	\node [squarednode] (R) at (2, 0) {$R$};
	\node [squarednode] (S) at (3, 0) {$S$};
	\node [squarednode] (y) at (2, 1) {$y$};
	\node [squarednode] (T) at (4, 1) {$T$};
	\node [squarednode] (x) at (3, 2) {$x$};
	
	\draw (x) to (y);
	\draw (x) to (T);
	\draw (y) to (U);
	\draw (y) to (R);
	\draw (y) to (S);
	
	\node [left=2cm of x] {$\tau^{c}_{Q_1}$:};
\end{tikzpicture}
\begin{tikzpicture}[-,>=stealth',roundnode/.style={circle,draw,inner sep=1.2pt},squarednode/.style={rectangle,inner sep=3pt}]
	\node [squarednode] (U) at (1, 0) {$U_{xy}$};
	\node [squarednode] (R) at (2, 0) {$R_{xyz}$};
	\node [squarednode] (R2) at (3, 0) {$R_{xyv}$};
	\node [squarednode] (x) at (2, 1) {$x$};
	
	\draw (x) to (U);
	\draw (x) to (R);
	\draw (x) to (R2);

	\node [left=1cm of x] {$\tau^{c}_{Q_2}$:};
\end{tikzpicture}
\begin{tikzpicture}[-,>=stealth',roundnode/.style={circle,draw,inner sep=1.2pt},squarednode/.style={rectangle,inner sep=3pt}]
	\node [squarednode] (U) at (1, 0) {$U_{xy}$};
	\node [squarednode] (R) at (2, 0) {$R_{xyz}$};
	\node [squarednode] (R2) at (3, 0) {$R_{xyv}$};
	\node [squarednode] (y) at (2, 1) {$y$};
	
	\draw (y) to (U);
	\draw (y) to (R);
	\draw (y) to (R2);

	\node [left=1cm of y] {$\tau'^{c}_{Q_2}$:};
\end{tikzpicture}

 		\caption{Example of \cqt{s} for $\tau_{Q_1}$, $\tau_{Q_2}$ and $\tau_{Q_2}'$.
		Since we have no self joins in $Q_1$, we omit the variables when representing the atoms.}
		\label{fig-compact-tree}
	\end{figure}
	
	We can now proceed with the proof. From $Q$ we construct the \acrocea:
	\[
	\ACH_Q = (I(Q) \cup \{\bar{x}\}, \uncq, \bincq, I(Q), \Delta, \{\qtree(\varepsilon)\})
	\]
	where the states of the automaton are the nodes of $\qtree$, the labeling set is the set off atom identifiers of the query $Q$, and the only final state is the variable at root of $\qtree$ (i.e., $\qtree(\varepsilon) = \qtree(\rt{\qtree})$). 
	For defining the transition relation $\Delta$, we need additional definitions regarding predicates and some further constructs over~$\qtree$.
	
	Below, we use $h: \{\bar{x}\} \cup \Data \rightarrow \Data$ to denote \changed{a homomorphism} and $\Hom$ to denote the set of all homomorphisms.
	Let $x \in \{\bar{x}\}$ be a variable and let $S(\bar{y}), T(\bar{z}) \in U(Q)$ be atoms of $Q$. We define the \emph{unary predicate of $S(\bar{y})$} as:
	\[
	U_{S(\bar{y})} \ := \  \{S(\bar{b}) \in \tuples[\Schema] \mid \exists h \in \Hom.\; h(S(\bar{y})) = S(\bar{b})\}
	\]
	and the \emph{binary predicate of $S(\bar{y})$ and $T(\bar{z})$} as:
	\[
	B_{S(\bar{y}), T(\bar{z})} \ := \ \{(S(\bar{b}), T(\bar{c})) \mid \exists h \in \Hom.\; h(S(\bar{y})) = S(\bar{b}) \wedge h(T(\bar{z})) = T(\bar{c})\}.
	\]
	One can easily check that $U_{S(\bar{y})} \in \uncq$ and $B_{S(\bar{y}), T(\bar{z})}$ is an equality predicate for every pair of atoms $S(\bar{y})$ and $T(\bar{z})$.  

	We extend the definition of $B_{S(\bar{y}), T(\bar{z})}$ to a pair variable-atom as follows. Take again the atom $S(\bar{y})$ and let $x$ be a variable \changed{such that  $x \notin \{\bar{y}\}$}. Note that this implies that $\{\bar{y}\} \cap \{\bar{x}_i\} = \{\bar{y}\} \cap \{\bar{x}_j\}$ for every pair of distinct atoms $i,j \in \desc{\qtree}{x}$. In other words, all atoms (i.e., identifiers) in the \qt below $x$ share the same variables with $S(\bar{y})$. 
	Define the \emph{binary predicate of $x$ and $S(\bar{y})$} as:
	\[
	B_{x, S(\bar{y})} = \bigcup_{i \in \desc{\qtree}{x}} B_{R_i(\bar{x}_i), S(\bar{y})}.
	\]
	Given the previous discussion and given that $Q$ does not have self joins (all atoms are different), one can easily check that $B_{x, S(\bar{y})}$ is an equality predicate in $\bincq$.
	
	For our last definition, let $i \in I(Q)$ be an identifier of $Q$ and $x \in \{\bar{x}_i\}$. Define the set of \emph{incomplete states of $x$ given $i$} as:
	\[
	C_{x, i} \ :=  \ \{\ell \in I(Q) \cup \{\bar{x}\} \mid \parent{\qtree}{\ell} \in \left(\desc{\qtree}{x} \cap \{\bar{x}_i\} \right)\} \setminus \big(\{i\} \cup \{\bar{x}_i\}\big)
	\]
	In other words, $C_{x, i}$ is the set of all variables or atom identifiers that hang from a variable of $\bar{x}_i$ that is a descendant of $x$, except for variables in $\bar{x}_i$ or the atom $i$. For example, in $\tau_{Q_1}$ of Figure~\ref{fig-compact-tree} one can check that $C_{y, U} = \{R, S, T\}$ and $C_{x,T} = \{y\}$.
	
	With these definitions, we define the transition relation $\Delta$ as follows:
	\[
	\begin{aligned}
		\Delta
		& \ := \ \Big\{\Big(\emptyset, U_{R_i(\bar{x}_i)}, \emptyset, \{i\}, i \Big) \mid\, i \in I(Q) \Big\}\\
		& \ \cup \ \Big\{\Big(C_{x,i}, U_{R_i(\bar{x}_i)}, \binfunc_{x,i}, \{i\}, x \Big) \mid i \in I(Q) \wedge x \in \{\bar{x}_i\}\Big\}
	\end{aligned}
	\]
	such that $\binfunc_{x,i}$ is the predicate function associated with $C_{x,i}$ defined as follows: $\binfunc_{x,i}(y) = B_{y, R_i(\bar{x}_i)}$ for every variable $y \in C_{x,i}$ and $\binfunc_{x,i}(j) = B_{R_j(\bar{x}_j), R_i(\bar{x}_i)}$ for every identifier $j \in C_{x,i}$.
	Intuitively, $\Delta$ can be interpreted as the traversal and completion of the \cqt, with the first set of transitions representing the independent leaf nodes being completed by their respective atom, while the second set represents the completion of a variable, meaning the last tuple in its descendants was found.
	
	From the previous construction, one can easily check by inspection that $\ACH_Q$ is of \changed{quadratic} size with respect to $Q$. In the sequel, we prove that $\ACH_Q$ is unambiguous and $\ACH_Q \equiv Q$. To prove both of these conditions, we will use the following lemma:
	\begin{lemma}\label{lemma:atom-on-run}
		For every accepting run $\rho$ of $\ACH_Q$ and for every atom $i\in I(Q)$, there is exactly one node $\bar{u} \in \rho$ such that $\rho(\bar{u}) = (p, j, \{i\})$.
	\end{lemma}
	\begin{proof}
		Let $i_R\in I(Q)$ be an atom with $Q(i_R) = R(\bar{x})$. Since $\rho$ is an accepting run of $\ACH_Q$, its root must be labeled by the root of the \cqt, $\rho(\rt{\rho}) = (r, j, \{i_R\})$ with $r = \rt{\qtree}$ and $i_S \in I(Q)$ such that $Q(i_S) = S(\bar{y})$. Assuming a non trivial case, i.e. $i_R \neq i_S$, there must be a transition $(C_{r, i_S}, U_{S(\bar{y})}, \binfunc_{r, i_S}, \{i_S\}, r) \in \Delta$, which leaves us with two recursive~cases.
		
		\changed{First, the case where none of the ancestors of $i_r$ are present in $C_{r, i_S}$, which will be the base case. Looking at the definition of $C_{r, i_S}$ and considering that $i_S \neq i_R$, this means that $i_R$ must be present in $C_{r, i_S}$ and therefore, there must be a node labeled by $(i_R, k, \{i_R\})$.
		
		On the other hand, if one of the ancestors of $i_R$ is present in $C_{r, i_S}$ the case will be recursive. Let this ancestor be $s \in \ancst{\qtree}{i_R} \cap C_{r, i_S}$ and since it is present in $C_{r, i_S}$, there must be a transition $(C_{s, i_T}, U_{i_T}, P_{i_T, C}, \{i_T\}, s) \in \Delta$, leaving us with the same three cases as before; the trivial case where $i_T = i_R$, the base case where there are no ancestors of $i_R$ in $C_{s, i_T}$ and the current recursive case. If we repeat the recursive case, we will eventually visit every ancestor of $i_R$, leaving us with either the trivial or the base case, meaning there exists a node $\bar{u} \in \rho$ such that $\rho(\bar{u}) = (p, j, \{i_R\})$.}

		We have proven that for every atom $i_R \in I(Q)$, there is a node $\bar{u}_R \in \rho$ such that $\rho(\bar{u}_R) = (p_R, j_R, \{i_R\})$ and $Q(i_R) = R(\bar{x}_R)$. Following the definition $\Delta$, one can easily see that if $\bar{u}_R$ is an inner node, $i_R \in \desc{\qtree}{p_R}$ and if $\bar{u}_R$ is a leaf node, $p_R = i_R$.
		
		On the other hand for each inner node, there must exist a transition of the form:
		\[
		(C_{p_R, i_R}, U_{R(\bar{x}_R)}, \binfunc_{p_R, i_R}, \{i_R\}, p_R) \in \Delta
		\] 
		and from the definition of $C_{p_R, i_R}$ it follows that for each $S(\bar{x}_S) \in \desc{\qtree}{p_R}$, with $S(\bar{x}_S)\neq R(\bar{x}_R)$ and $Q(i_S) = S(\bar{x}_S)$, there is exactly one node $\bar{u}_S \in \rho$ such that $\rho(\bar{u}_S) = (i_S, j_S, \{i_S\})$ or $\rho(\bar{u}_S) = (p_S, j_S, \{i\})$ where if $i \neq i_S$, $i_S \in \desc{\qtree}{p_S}$. Since each transition of $\Delta$ has exactly one atom as its label, then there will be exactly one node in the run marked by each atom.
	\end{proof}

	The automaton $\ACH_Q$ will be unambiguous iff every accepting run of $\ACH_Q$ is simple and for every stream $\Stream$ and a valuation $\nu$, there is at most one run $\rho$ of $\ACH_Q$ over $\Stream$ such that $\nu_\rho = \nu$. Every transition of $\ACH_Q$ uses a single label $i \in I(Q)$, which combined with the correspondence between atoms and nodes in each run of directly indicates that every accepting run of $\ACH_Q$ is~simple.
	
	On the other hand, it is easy to see that given a tuple $R(\bar{a}) \in \Stream$, there will only be a single initial transition $(\emptyset, U_{R_i(\bar{x}_i)}, \emptyset, \{i\}, i ) \in \Delta$ such that $R(\bar{a}) \in U_{R_i(\bar{x}_i)}$ and given a set of possible states for a run of $\ACH_Q$, $C$ there will be at most one transition, due to the definition of $C_{x,i}$, that can be taken by $\ACH_Q$, meaning there will be at most one accepting run for each corresponding output and therefore $\ACH_Q$ will be unambiguous.
	
	Now proving that $\ACH_Q \equiv Q$, this will hold iff $\sem{\ACH_Q}_n(\Stream) = \sem{Q}_n(\Stream)$ for every position $ n\in\nat $ and every stream $\Stream$, which means we need to prove the following two conditions:
	\begin{itemize}
	\item $\sem{\ACH_Q}_n(\Stream) \subseteq \sem{Q}_n(\Stream)$.
		The output of the automaton $\ACH_Q$ over the stream $\Stream=t_0t_1\ldots$ at position $n$ is defined as the set of valuations:
		\[
		\sem{\ACH_Q}_n(\Stream) \ = \  \{\nu_\rho \mid \text{$\rho$ is an accepting run of $\ACH_Q$ over $S$ at position $n$}\}
		\]
		meaning that for every valuation $\nu_\rho \in \sem{\ACH_Q}_n(\Stream)$ there is an associated accepting run tree $\rho: t \rightarrow (2^\Omega \setminus \{\emptyset\})$.

		Because of Lemma~\ref{lemma:atom-on-run}, we know that for every atom $i_R\in I(Q)$, with $Q(i_R) = R(\bar{x}_R)$, there is a node $\bar{u}_R \in \rho$ such that $\rho(\bar{u}_R) = (p_R, j_R, \{i_R\})$. Using this, we can define $\eta: I(Q) \to I(D_n[\Stream])$ and $h_\eta: X \cup D_n[\Stream] \to D_n[\Stream] $ such that for every $i_R \in I(Q)$, $\eta(i_R) = j_R$ and:
		\[ h_\eta(x) = \begin{cases}
			x & x \in D_n[\Stream] \\
			D_n[\Stream](j_R) & x \in \Var \\
		\end{cases}
		\]
		
		Since $h_\eta(a) = a$ for every $a \in D_n[\Stream]$, it is \changed{a homomorphism}, so if $h_\eta$ is well-defined it will be \changed{a homomorphism} from $Q$ to $D_n[\Stream]$ and $\eta$ will be a t-homomorphism from $Q$ to $D_n[\Stream]$ with $h_\eta$ as its corresponding homomorphism.
		
		Consider the atoms $i_R, i_T \in I(Q)$ and the nodes of the run tree $\bar{u}_R, \bar{u}_T \in \rho$ such that $\rho(\bar{u}_R) = (p_R, j_R, \{i_R\})$, $\rho(\bar{u}_T) = (p_T, j_T, \{i_T\})$, $Q(i_R) = R(\bar{x}_R)$ and $Q(i_T) = T(\bar{x}_T)$. We need to prove that for every $x\in\Var$, $h_\eta(x)$ has a single value; in other words, if $z \in \{\bar{x}_R \cap \bar{x}_T\}$, then $D_n[\Stream](j_R).z = D_n[\Stream](j_T).z$. Given the structure of the \qt, the path from the root to every atom corresponds to the variables of the atom, so $z \in \ancst{\qtree}{R(\bar{x}_R)} \cap \ancst{\qtree}{T(\bar{x}_T)}$. 
		
		For the purpose of simplification we will start by assuming that one of the states of the run is marked by $z$, and then explain the general case. Let $i_S \in I(Q)$ be an atom and $\bar{u}_S \in \rho$ be a node of the run tree such that $\rho(\bar{u}_S) = (z, j_S, \{i_S\})$ and $Q(i_S) = S(\bar{x}_S)$, which means there must be a transition of the form, $(C_{z,i_S}, U_{S_i(\bar{x}_S)}, \binfunc_{z,i_S}, \{i_S\}, z)$ in $\Delta$. Given $z \in \ancst{\qtree}{R(\bar{x}_R)}$, we have two possible cases.
		
		Case (1) corresponds to $p_R \in C_{z,i_S}$. This implies that the tuples associated with $i_R$ and $i_S$ must satisfy the binary predicate given by the transition, i.e. $\big(D_n[\Stream](j_R), D_n[\Stream](j_S)\big) \in \binfunc_{z, i_S}(i_R))$, and therefore, by the definition of the binary predicate of $R(\bar{x}_R)$ and $S(\bar{x}_S)$, $D_n[\Stream](j_R).z = D_n[\Stream](j_S).z$.
		
		On the other hand, for case (2), $p_R \notin C_{z,i_S}$. Since there is a node labeled by $z$ in the run tree, the variable must have been completed during the run, meaning there must be a sequence of nodes $\bar{u}_1,\ldots,\bar{u}_m \in \rho$, with $\bar{u}_1 = \bar{u}_S$, $\bar{u}_m = \bar{u}_R$, such that for each node $\bar{u}_k$, $\rho(\bar{u}_k) = (p_k, j_k, \{i_k\})$ and $\bar{u}_{k} = \parent{\rho}{\bar{u}_{k+1}}$. This means that $p_{k+1} \in C_{p_k, i_k}$ and $p_{k+1} \in \desc{\qtree}{p_k}$ for every node $\bar{u}_k$, which indicates that if $p_k \neq z$, then $p_k \in \desc{\qtree}{z}$ and $z$ is a variable of $Q(i_k)$. Following the same steps of (1), it is easy to see that for every $\bar{u}_k$, $\big(D_n[\Stream](j_{k+1}), D_n[\Stream](j_{k})\big) \in \binfunc_{p_k, i_k}(i_{k+1})$ and $D_n[\Stream](j_{k}).z = D_n[\Stream](j_{k+1}).z$, therefore $D_n[\Stream](j_R).z = D_n[\Stream](j_S).z$. Note that we purposefully omitted the case $i_S = i_R$ in which the condition holds trivially.

		It is clear that the same arguments as in (1) and (2) are valid for the atom $i_T$, so if the run tree has a node labeled by $z$, it holds that $D_n[\Stream](j_T).z = D_n[\Stream](j_S).z = D_n[\Stream](j_R).z$. If there is no such node in the run, we can start from the root of the run and its corresponding transition, where $\rho(\rt{\rho}) = (r, j_S, \{i_S\})$ with $r = \rt{\qtree}$ and $(C_{r, i_S}, U_{S(\bar{x}_S)}, \binfunc_{r, i_S}, \{i_S\}, r) \in \Delta$ with $Q(i_S) = S(\bar{x}_S)$. It is easy to see that if there are no nodes labeled by $z$, then $z$ must be a variable of $S(\bar{x}_S)$, so we can use the exact same arguments that we used in (1) and (2) to prove that $D_n[\Stream](j_R).z = D_n[\Stream](j_S).z$ and $D_n[\Stream](j_T).z = D_n[\Stream](j_S).z$, therefore in any case it holds that $D_n[\Stream](j_T).z = D_n[\Stream](j_R).z$ and $h_\eta$ is well-defined and \changed{a homomorphism} from $Q$ to $D_n[\Stream]$.

		Finally, since we have a t-homomorphism from $Q$ to $D_n[\Stream]$ for every run tree $\rho$ of $\ACH_Q$, and $\sem{Q}_n(\Stream) \ = \ \{\hat{\eta} \mid \text{$\eta$ is a t-homomorphism from $Q$ to $D_n[\Stream]$}\}$, then $\sem{\ACH_Q}_n(\Stream) \subseteq \sem{Q}_n(\Stream)$.

	\item $\sem{Q}_n(\Stream) \subseteq \sem{\ACH_Q}_n(\Stream)$.
		Let $\sem{Q}_n(\Stream) \ = \ \{\hat{\eta} \mid \text{$\eta$ is a t-homomorphism from $Q$ to $D_n[\Stream]$}\}$ be the set of outputs of $Q$ over the stream $\Stream$.
		
		Let $\hat{\eta} \in \sem{Q}_n(\Stream)$ be a valuation with its associated t-homomorphism $\eta$ and homomorphism from $h_\eta$. We can represent the tuples of $\hat{\eta}$ as the sequence $t_\eta = \bleft i_1, \ldots i_m \bright$, where for all $k \in t_\eta$, $ D_n[\Stream](i_k) = R_k(\bar{a}_k)$ and $i_{k} < i_{k+1}$.

		Using the valuation and the sequence $t_\eta$, we can build a run tree of $\ACH_Q$, $\rho: t \rightarrow (Q \times \nat \times (2^\Omega \setminus \{\emptyset\}))$. We starting with a single leaf node $\varepsilon$, with $\rho(\varepsilon) = (R_1(\bar{x}_1), t_\eta(1), \{ R_1(\bar{x}_1) \})$. After this, for every tuple $k \in t_\eta$, we must check if there is a transition $(P, U, B, L, q) \in \Delta$ such that for each $p \in P$, there is a node $\bar{u} \in \rho$ with $\rho(\bar{u}) = (p, j, L)$ and $\bar{u} \notin \children{\rho}{\bar{v}}$ for every $\bar{v}\in\rho$, i.e. we check if there are nodes with no parent that are labeled by each state needed for a transition. If we find a transition that satisfies the previous condition, we add a new node $\bar{v}$ labeled by $(q, t_\eta(k), L)$; if there are no transitions that satisfy the condition, we add a new leaf node labeled by $(R_k(\bar{x}_k), t_\eta(k), \{ R_k(\bar{x}_k) \})$.

		Since we get each tuple from the values given by the homomorphism $h_\eta$, we know that they will satisfy the unary and binary predicates of each transition, and since there is exactly one tuple for each atom of $Q$, we know that there will be a one to one correspondence between the tuples and the labels in the tree $\rho$, meaning it is an accepting run and therefore $\sem{Q}_n(\Stream) \subseteq \sem{\ACH_Q}_n(\Stream)$.
	\end{itemize}

	\paragraph{Connected HCQ with self joins} For a non-empty set $A \subseteq I(Q)$ we say that \emph{$A$ is a self join} (of $Q$) iff $R_i = R_j$ for every $i, j \in A$. That is, if all relations of the identifiers are the same. We denote by $\SJ{Q}$ the set of all self joins $A$ of $Q$. Note that $\bigcap_{i \in A} \{\bar{x}_i\} \neq \emptyset$ for every $A \in \SJ{Q}$ since $\qtree(\varepsilon) \in \bigcap_{i \in A} \{\bar{x}_i\}$ (i.e., $Q$ is connected).
	
	To deal with self joins, we must modify the previous construction by keeping track of the last atom or self join that was read. This modification is necessary to use equality predicates. Furthermore, the construction has an inherently exponential blow-up in the number of transitions, given that we need to annotate tuples with self joins. This blow-up seems unavoidable for the model since, if $Q$ has only atoms with the same relational symbol, then the last transitions is forced to annotate with sets of atoms, which are an exponential number of them. 
	
 	Let $\VSJ{Q} = \{(x, A) \mid A \in \SJ{Q} \wedge\ x \in \bigcap_{i \in A} \{\bar{x}_i\}\}$. We define the \acrocea:
	\[
	\ACH_Q = (I(Q) \cup \VSJ{Q}, \uncq, \bincq, I(Q), \Delta, \{\qtree(\varepsilon)\} \times \SJ{Q}).
	\] 
	Note that $\{\qtree(\varepsilon)\} \times \SJ{Q} \subseteq \VSJ{Q}$. The main change is that the variable states of the automaton are pairs formed by a variable and a self join that brought $\ACH_Q$ to that variable.
	We dedicate the rest of this subsection to introduce the notation that we need to define $\Delta$.
	
	To account for a single tuple satisfying a self join of $Q$, we need to define a new unary predicate associated with a self join $A \in \SJ{Q}$. For this, the following lemma will be relevant. 
	\begin{lemma}\label{lemma:single-atom}
		Let $A \in \SJ{Q}$. There exists an atom $t_A$ (not necessarily in $Q$) such that, for every $t \in \tuples[\Schema]$,  the following statements are equivalent: (1) there exists $h \in \Hom$ such that $h(\exatom{i}) = t$ for every $i \in A$; (2) there exists $h' \in \Hom$ such that $h'(t_A) = t$. 
	\end{lemma}
	\begin{proof}
		Let $[1, n] = \{1, \ldots, n\}$ and define the relation $H \subseteq [1, n]^2$ such that:
		\[
			H = \{(k_1, k_2 \in [1, n]^2 \mid \exists i, j.\, \bar{x}_i[k_1] = \bar{x}_j[k_2])\}
		\]
		Note that $H$ is reflexive and symmetric, but not necessarily transitive.
		
		Let $H^T$ be the transitive closure of $H$. Now $H^T$ is an equivalence relation. Given $k\in[1, \ldots, n]$, let $[k]$ be the equivalence class of $k$ with respect to $H^T$. Define $R(\bar{x}) = R([1], \ldots, [n])$. Next we prove that the atom $t_A = R(\bar{x})$ satisfies the theorem.
		\begin{itemize}
			\item (1) $\rightarrow$ (2). Suppose that $h$ is \changed{a homomorphism} such that $h(R(\bar{x})) = R(\bar{a})$. Define the homomorphism $h^\ast(x) = [k]$, with $\bar{x}_i[k] = x$ for some i.
			
			One can prove that (1) $h^\ast$ is well-defined, i.e., it does not depend on $k$ and (2) $h^\ast(R(\bar{x}_i)) = R(\bar{x})$ for all $i \leq n$.
			
			Define $h' = h^\ast \circ h$, then $h'(R(\bar{x}_i)) = h(h^\ast(R(\bar{x}_i))) = h(R(\bar{x})) = R(\bar{a})$.
			
			\item (2) $\rightarrow$ (1). Suppose that $h'\in\hom$ satisfies  $h'(R(\bar{x}_i)) = R(\bar{a})$ for every $i$.
			
			Define the homomorphism $h$ such that $h([k]) = h'(\bar{x}_1[k])$. It is easy to see that $h$ is well-defined, namely if $[k_1] = [k_2]$, then $h'(\bar{x}_1[k_1]) = h'(\bar{x}_1[k_2])$.

			Since $h$ is well-defined, then $h$ is \changed{a homomorphism} and:
			\[
			\begin{aligned}
				h(R(\bar{x}))
				&= R(h([1]), \ldots, h([n])) \\
				&= R(h'(\bar{x}_1[1]), \ldots, h'(\bar{x}_1[n])) \\
				&= h'(R(\bar{x}_1)) \\
				&= R(\bar{a}) \\
			\end{aligned}
			\]
		\end{itemize}
	\end{proof}
	
	Note that the homomorphism $h$ implicitly forces that $t_A$ has the same relational symbol as every $R_i$ with $i \in A$, otherwise there would be no tuple that satisfies the single homomorphism condition. It also motivates the following unary predicate for a self join $A \in \SJ{Q}$: 
	\[
	U_{A} = \{t \in \tuples[\Schema] \mid \exists h \in \Hom. \, h(t_A) = t\}.
	\]
	One can check that $U_{A} \in \uncq$ given that one can compute $t_A$ from $A$ in advance, and then check that $t$ is homomorphic to $t_A$ in linear time over $t$ for every $t \in \tuples[\Schema]$.
	
	Similar to unary predicates, we can derive a lemma for a pair of self joins.
	
	\begin{lemma}\label{lemma:double-atoms}
		Let $A_1, A_2 \in \SJ{Q}$. There exist atoms $\cev{t}_{A_1, A_2}$ and $\vec{t}_{A_1, A_2}$ (not necessarily in~$Q$) such that, for every pair $(t_1, t_2) \in \tuples[\Schema]^2$, the following statements are equivalent: (1)~there exists $h \in \Hom$ such that $h(\exatom{i}) = t_j$ for every $j \in \{1,2\}$ and every $i \in A_j$; (2)~there exists $h' \in \Hom$ such that $h'(\cev{t}_{A_1, A_2}) = t_1$ and $h'(\vec{t}_{A_1, A_2}) = t_2$. 
	\end{lemma}
	\begin{proof}
		Let $A_1 = R_1(\bar{x}_1, \ldots, R_1(\bar{x}_m))$, $A_2 = R_2(\bar{y}_1, \ldots, R_2(\bar{y}_{m'}))$, $[1, n] = \{1, \ldots, n\}$ and $[1', n'] = \{1', \ldots, n'\}$. Note that $[1, n] \cap [1', n'] = \emptyset$. Define the relation $H \subseteq ([1, n] \cup [1', n'])^2$ such that:
		\[
		\begin{aligned}		
			H
			&= \{(k_1, k_2 \mid \exists i, j.\, \bar{x}_i[k_1] = \bar{x}_j[k_2])\}\\
			&= \{(k_1', k_2' \mid \exists i, j.\, \bar{y}_i[k_1] = \bar{y}_j[k_2])\}\\
			&= \{(k_1, k_2' \mid \exists i, j.\, \bar{x}_i[k_1] = \bar{y}_j[k_2])\}\\
			&= \{(k_1', k_2 \mid \exists i, j.\, \bar{y}_i[k_1] = \bar{x}_j[k_2])\}\\
		\end{aligned}
		\]
		Just like for Lemma~\ref{lemma:single-atom}, take $H^T$ the transitive closure of $H$, with $H^T$ an equivalence relation. Given $k\in[1, \ldots, n] \cup [1', n']$, let $[k]$ be the equivalence class of $k$ in $H^T$. Define $R_1(\bar{x}) = R_1([1], \ldots, [n])$ and $R_2(\bar{y}) = R_2([1'], \ldots, [n'])$. The proof that $\cev{t}_{A_1, A_2}$ and $\vec{t}_{A_1, A_2}$ is completely analogous to the one in Lemma~\ref{lemma:single-atom}.

	\end{proof}
	Analog to the unary predicates, this lemma motivates the following binary predicate for every pair $A_1, A_2 \in \SJ{Q}$: 
	\[
	B_{A_1,A_2} = \{(t_1,t_2) \in \tuples[\Schema]^2 \mid \exists h \in \Hom. \, h(\cev{t}_{A_1, A_2}) = t_1 \wedge h(\vec{t}_{A_1, A_2}) = t_2\}.
	\]
	Note that $B_{A_1,A_2} \in \bincq$. Indeed, $B_{A_1,A_2} = B_{S(\bar{y}), T(\bar{z})}$ with $S(\bar{y}) = \cev{t}_{A_1, A_2}$ and $T(\bar{z}) = \vec{t}_{A_1, A_2}$.

	Now, we move to the definitions for constructing the transition relation $\Delta$. Let $A \in \SJ{Q}$. For a variable $x \in \bigcap_{i \in A} \{\bar{x}_i\}$ we define the set of \emph{incomplete states of $x$ given $A$} as:
	\[
	C_{x, A} \ :=  \ \big\{\ell \in I(Q) \cup \{\bar{x}\} \mid \parent{\qtree}{\ell} \in \big(\desc{\qtree}{x} \cap \bigcup_{i \in A} \{\bar{x}_i\} \big)\big\} \setminus \big(A \cup \bigcup_{i \in A} \{\bar{x}_i\}\big).
	\]
	$C_{x, A}$ is the generalization of the set $C_{x, i}$. Indeed, one can check that $C_{x, i} = C_{x,\{i\}}$. Then, the intuition behind $C_{x, A}$ is similar: in $C_{x,A}$ we are collecting all variables or atoms identifiers in $\qtree$ that directly hangs from a variable that it is a descendant of $x$ and it is in an atom of $A$, except for the same variables and atoms in $A$.
	
	Let $C_{x, A}$ be the incomplete states of $x$ given $A$. We say that a subset of states $C \subseteq I(Q) \cup \VSJ{Q}$ \emph{encodes} $C_{x, A}$ in $\ACH_Q$ iff $C \cap I(Q) = C_{x, A} \cap I(Q)$ and for every variable $y \in C_{x,A}$ there exists a unique pair $(y, A') \in C$. Intuitively, $C$ will be the analog of $C_{x,i}$ that we use as a set of states in the case without self joins. Note that $C_{x,A}$ can be empty, and then $\emptyset$ is the only set representing $C_{x,A}$. We denote by $\bar{C}_{x, A}$ the set of all encodings of $C_{x,A}$ in $\ACH_Q$. 

	With this machinery, we construct the new transition relation $\Delta$ as follows:
	\[
	\begin{aligned}
		\Delta
		& = \Big\{\Big(\emptyset, U_{R_i(\bar{x}_i)}, \emptyset, \{i\}, i \Big) \mid\, i \in I(Q) \Big\}\\
		& \cup \Big\{\Big(C, U_{A}, \binfunc_{C, A}, A, (x, A) \Big) \mid A \in \SJ{Q} \wedge \, x \in \bigcap_{i \in A} \{\bar{x}_i\} \wedge  C \in \bar{C}_{x, A} \Big\}
	\end{aligned}
	\]
	such that $\binfunc_{C, A}: C \rightarrow \bincq$ is the predicate function associated with $C$ and $A$ defined as follows: $\binfunc_{C, A}(j) = B_{\{j\}, A}$ for every identifier $j \in C$, and $F_{x,C}((y,A')) = B_{A', A}$ for every $(y,A') \in C$. Note that all the binary predicates of $F_{C,A}$ are equality predicates as defined above. 
	
	This ends the definition of $\Delta$ and $\ACH_Q$ for a connected HCQ $Q$ with self joins. Next, we prove that $\ACH_Q$ is unambiguous and $\ACH_Q \equiv Q$.

	Thanks to Lemma~\ref{lemma:single-atom} and Lemma~\ref{lemma:double-atoms}, this proof is analogous the the one without self joins.

	\paragraph{The general case} We end this proof by showing the case when $Q$ is disconnected. If this is the case, let $x^*$ be a fresh variable and redefine $Q$ as:
	\[
	Q^*(x^*, \bar{x}) \ \leftarrow \ R_0(x^*,\bar{x}_0), \ldots, R_{m-1}(x^*, \bar{x}_{m-1})
	\]
	where $Q^*$ is defined over a new schema $\sigma^*$ where the arity of each relational symbol is increased by one. $Q^*$ is hierarchical and now is connected. By the previous case, we can construct a \acrocea $\ACH_{Q^*}$ such that $\ACH_{Q^*} \equiv Q^*$. Now, take $\ACH_{Q^*}$ and define $\ACH_{Q}$ by removing the fresh variable $x^*$ from $\ACH_{Q^*}$, namely, remove it from the unary and binary predicates in the transitions. One can easily check that $\ACH_Q \equiv Q$.  
	\end{proof}

\subsection*{Proof of Theorem~\ref{theo:hierarchical-onlyif}}

	Fix a schema $\sigma$, a set of data values $\Data=\nat$ and an acyclic CQ $Q$ over $\sigma$ of the form:
	\[
		Q(\bar{x}) \ \leftarrow \ R_0(\bar{x}_0), \ldots, R_{m-1}(\bar{x}_{m-1})
	\]
	that \textbf{is not} a HCQ, meaning there is a pair of variables $y, z \in \Var$ with $\atoms{y} \not\subseteq \atoms{z}$, $\atoms{z} \not\subseteq \atoms{y}$ and $\atoms{y} \cap \atoms{z} \neq \emptyset$. Without loss of generality, assume that $y \in \{\bar{x_0}\} \cap \{\bar{x_1}\}$ , $y \notin \{\bar{x_2}\}$ and $z \in \{\bar{x_0}\}\cap\{\bar{x_2}\}$ and $z \notin \{\bar{x_1}\}$; moreover, $R_0(x_0^0, \ldots x_0^{n_0})$, $R_1(x_1^0, \ldots x_1^{n_1})$, $R_2(x_2^0, \ldots x_2^{n_2})$ with $x_0^0 = x_1^0 = y$ and $x_0^1 = x_2^0 = z$.

	Let $\Stream_k = \bleft\, R_0(\bar{a}_0), \ldots, R_{m-1}(\bar{a}_{m-1}), \ldots\,\bright$ be a family of streams where for every tuple $R_i(a_i^0, \ldots a_i^{n_i})$, $a_i^j = k$, with $k \in \nat$ if $x_i^j = y$ or $x_i^j = z$ and $a_i^j = 0$ if $x_i^j \neq y$ and $x_i^j \neq z$. In other words, every variable of every tuple is mapped to $0$, except for $y$ and $z$, \changed{which} are mapped to $k$. \changed{Recalling our previous definition for CQ over streams, i.e. considering $\Omega = I(Q)$, } it is clear that the valuation $\nu = \bleft 0, 1, 2, \ldots, m-1 \bright \in \sem{Q}_n(\Stream_k)$ for every $k \in \nat$.

	Assume there is a \anamecea $\ACH_Q = \big(\ACH_Q, \un, \bin, \Omega, \Delta, F \big)$ such that $\ACH_Q \equiv Q$, meaning that for every stream and, in particular, for every stream $\Stream_k$ it holds that $\sem{\ACH_Q} = \sem{Q}$, which means $\nu \in \sem{\ACH_Q}_n(\Stream_k)$ for every $k \in \nat$. Since $\ACH_Q$ has a finite \changed{number} of states, there must exist $j, k \in \nat$ such that their accepting runs for $\nu$, $\rho_j$ and $\rho_k$ respectively, are isomorphic, meaning their runs go through the exact same states. For the sake of simplification, we assume that $\dom(\rho_j) = \dom(\rho_k)$.

	Let $\Stream_{j \leftarrow k} = \bleft R_0(\bar{a}_0), \ldots, R_{m-1}(\bar{a}_{m-1}), \ldots\bright$ where for every tuple $R_i(a_i^0, \ldots a_i^{n_i})$, $a_i^\ell = j$ if $x_i^\ell = y$, $a_i^\ell = k$ if $x_i^\ell = z$ and $a_i^\ell = 0$ if $x_i^\ell \neq y$ and $x_i^\ell \neq z$. It is clear that this time, $\nu = \bleft 0, 1, 2, \ldots, m-1 \bright \notin \sem{Q}_n(\Stream_{j \leftarrow k})$.

	Since $\rho_j$ and $\rho_k$ are runs associated with $\nu$, there are nodes $\bar{u}_0, \bar{u}_1, \bar{u}_2 \in I(\rho_j)$ such that $\rho_j(\bar{u}_i) = (p_i, i, L_i)$ and $\rho_k(\bar{u}_i) = (p_i, i, L_i)$, with $R_i(\bar{x}_i) \in L_i$ for each $i\in\{0, 1, 2\}$. Since this nodes exist in $\rho_j$ and $\rho_k$, for every $i\in\{0, 1, 2\}$ there must transitions $(P_i, U_i, B_i, L_i, p_i)$ and $(P_i, U_i', B_i', L_i, p_i)$ such that $\Stream_j(i)$ satisfies the unary and binary predicates $U_i, B_i$ and $\Stream_k(i)$ satisfies the unary and binary predicates $U_i', B_i'$.
	
	Using these transitions, we can define $\rho_{j\leftarrow k}: I(\rho_j) \to (\ACH_Q, \nat, (2^\Omega \setminus \emptyset))$ with $\rho_{j\leftarrow k}(\bar{u}) = \rho_j(\bar{u})$ for every $\bar{u} \in \dom(\rho_{j\leftarrow k})$. One can easily check that $\rho_{j\leftarrow k}$ is an accepting run tree of $\ACH_Q$ over $\Stream_{j\leftarrow k}$, since it can follow $(P_i, U_i', B_i', L_i, p_i)$ for $R_0(\bar{a}_0), R_1(\bar{a}_1), R_2(\bar{a}_2)$ and follow the exact same transitions as $\rho_j$ for every other tuple.
	
	Since for the valuation $\nu = \bleft 0, 1, 2, \ldots, m-1 \bright$ it holds that $\nu \notin \sem{Q}_n(\Stream_{j \leftarrow k})$ and $\nu \in \sem{\ACH_Q}_n(\Stream_{j \leftarrow k})$, then $\ACH_Q \not\equiv Q$ and therefore if $Q$ is and acyclic CQ that is not hierarchical, then $\ACH_Q \not\equiv Q$ for all PFA $\ACH_Q$. 	
	\section{Proofs of Section~\ref{sec:algorithm}} \label{sec:app-algorithm}

\subsection*{Proof of Theorem~\ref{theo:enumeration}}
\begin{proof}
    
    Let $w \in \nat$ be a window size, $\DSw$ be a simple data structure and $\n\in\dsnodesw$ be a node of the data structure. The valuations in $\dssem{\n}^w_\ipos$ are defined as:
    \[
        \dssem{\n}^w_\ipos \ := \ \bleft \nu \in \dssem{\n} \mid |\ipos - \min(\nu)| \leq \window  \bright.
    \]
    with
    \[
        \dssemprod{\n} := \bleft\nu_{L(\n), i(\n)}\bright  \valop \bigvalop_{\n' \in \dsprod(\n)} \dssem{\n'}  \ \ \ \ \ \ \ \ \ \ 
        \dssem{\n} := \dssemprod{\n} \cup \dssem{\dsleft(\n)} \cup \dssem{\dsright(\n)}.
    \]

    Following the definitions used in \cite{MunozR22}, we will say that the algorithm enumerates the results $\nu \in \dssem{\n}^w_\ipos$ by writing $\#\nu_1\#\nu_2\#\ldots\#\nu_m\#$ to the output registers, where $\#\notin\Omega$ is a separator symbol. Let $\outtime(i)$ be the time in the enumeration when the algorithm writes the $i$-th symbol $\#$, we define the $\outdelay(i) = \outtime(i+1) - \outtime(i)$ for each $i \leq m$.
    We say that the enumeration has \emph{output-linear delay} if there is a constant $k$ such that for every $i\leq m$ it holds that $\outdelay(i) \leq k \cdot |\nu_i|$.

    To output the first valuation of $\dssem{\n}^w_\ipos$  we need to (1) determine if $\dssem{\n}^w_\ipos = \emptyset$ and (2) build the valuation by calculating the products in $\dssem{\n}$. We can know that $\dssem{\n}^w_\ipos \neq \emptyset$ iff $|\ipos - \dsmaxstart(\n)| \leq \window$, and since the value of $\dsmaxstart(\n) = \max\{\min(\nu) \mid \nu \in \dssemprod{\n} \}$ is stored in every node $\n$ and we are doing a simple calculation with constants, we can check (1) in constant time. Note that it is not necessary to recursively check the $\dsmaxstart$ of the rest of the nodes in $\dssemprod{\n}$ since they are considered in the definition.
    
    On the other hand, the product of two bags of valuations $V, V'$ is defined as the bag $V \valop V' = \bleft\nu \valop \nu' \mid \nu \in V, \nu' \in V' \bright$, where $\nu \valop \nu'$ is the product of two valuations, defined as a valuation such that $[\nu \valop \nu'](\ell) = \nu(\ell) \cup \nu'(\ell)$ for every $\ell \in \Omega$. With these definitions, we can enumerate a single valuation $\nu \in \dssemprod{\n}$ by calculating the union between a valuation $\nu_\n \in U(\bleft\nu_{L(\n), i(\n)}\bright)$ and $\nu_{\n'} \in \dsprod(\n')$ for each $n' \in \dsprod(\n')$. It is easy to see that we can complete (2) by both calculating and writing this valuation in linear time. It is worth noting that we can make sure that we find valuations inside of the time window in constant time by traversing every bag in reverse order (starting from the valuations with a higher to lower~$\min\{\nu\}$).

    After enumerating the first output, we can continue traversing the bags of valuations, checking in constant time if $|i - \min\{\nu_{\n'}\} \leq w|$. In the worst case, which will occur right after writing the last valuation in the output, we will have to check that $|i - \min\{\nu_{\n'}\} \leq w|$ for every node $\n' \in \dsprod(n)$, but since each check takes constant time and there is one node for each valuation we are adding to the output, this step can also be done in linear time with respect to $|\nu|$. Finally, after enumerating every output of $\dsprod(n)$ inside the time window, we can recursively start the enumeration for $\dsleft(\n)$ and $\dsright(\n)$ in constant time, which will maintain an output-linear delay.
\end{proof}

\subsection*{Proof of Proposition~\ref{prop:union}}

\begin{proof}
    
    Fix $k, w\in\nat$ and assume that one performs $\dsunion(\n_1, \n_2)$ over $\DSw$ with the same position $i = \dspos(\n_2)$ at most $k$ times. In the following, we first prove the proposition with an implementation of the $\dsunion$ operation that is not fully persistent and then show how to modify the implementation to maintain this property.
    
    \changed{Let $\n_1, \n_2 \in \dsnodesw$ be two nodes such that $\max(\n_1) \leq \dspos(\n_2)$ and $\dsleft(\n_2) = \dsright(\n_2) = \bot$. We say that $\n_1 \leq \n_2$ iff (1) $\dsmaxstart(\n_1) \leq \dsmaxstart(\n_2)$ and (2) if $\dsmaxstart(\n_1) = \dsmaxstart(\n_2)$ then $i(\n_1) \leq i(\n_2)$.

    Recall that this operation requires inserting $\n_2$ into $\n_1$ and it outputs a fresh node $\n_u$ such that $\dssem{\n_u}^w_\ipos := \dssem{\n_1}^w_\ipos \cup \dssem{\n_2}^w_\ipos$.
    
    If $|\dsmaxstart(\n_1 - i(n_2))| > w$ then all of the outputs from $\n_1$ are now out of the time window, so $\dssem{\n_1}^w_\ipos \cup \dssem{\n_2}^w_\ipos = \dssem{\n_2}^w_\ipos$ and therefore $\dsunion(\n_1, \n_2) = \n_2$. The time it will take to do this operation will be the time necessary to insert $\n_2$ in $\DSw$, which we will analyze later. 
    
    On the other hand, if $|\dsmaxstart(\n_1 - i(\n_2))| \leq w$, we have to consider the outputs of both nodes, $\n_2$ and $\n_1$. First check how $\n_1$ compares with $\n_2$. If $\n_1 \leq \n_2$, then we have to create the new node $\n = \n_2$ such that $\dsleft(\n) = \dsunion(\dsleft(\n_2), \n_1)$ and, similarly, if $\n_1 > \n_2$, we need to create the new node $\n = \n_1$ such that $\dsleft(\n) = \dsunion(\dsleft(\n_1), \n_2)$}. In both cases, we are only creating one node and switching or adding pointers between nodes a constant number of times\changed{. Although} it might seem like a recursive operation at first glance, we know beforehand that $\dsmaxstart(\n) \geq \dsmaxstart(\dsleft(\n))$, so there will be at most one other union process generated. Once again since we can do this part of the operation in constant time \changed{and }the bulk of the operation will be the time necessary for the insertion of the new node.

    We know that for every position $i = \dspos(\n_2)$ we will perform a $\dsunion$ operation at most $k$ times. Starting with an empty data structure, there will be at most $k \cdot w$ nodes in $\DSw$ given a time window $w$. Assuming $\DSw$ is a perfectly balanced binary tree, this means that the tree has a depth of $\log_2(k\cdot w)$.
    
    To ensure that the tree will always be balanced, we can add one bit of information to every node, which we will call \emph{the direction bit}, that indicates which of the children of the node we need to visit for the insertion. If $bit(\n) = 0$, we must go to its left child and we must go to the right one otherwise. After each insertion, we need to change the value of the direction bit of every node in the path from the root to the newly inserted one, to avoid repeating the same path on the next insertion. This operation can be done in constant time for each node, so the time it will take to update all of the direction bits for each insertion will be exactly the depth of the tree.

    Since one performs $\dsunion(\n_1, \n_2)$ over $\DSw$ with the same position $i = \dspos(\n_2)$ at most $k$ times, if we start with an empty data structure, it will have at most $k \times w$ nodes after reading $w$ tuples from the stream. To insert the next node $\n'$, by following the direction bits, we will end up in the oldest node of the tree $\n$, but it is clear that $i(\n) \leq i(\n') + w $, meaning that $\dsmaxstart(\n) - i(\n) \leq w$ which indicates that all of the outputs of $\n$ are outside of the time window and therefore we can safely remove $\n$ from the tree and replace it with $\n'$ without losing outputs. Given that the depth of $\DSw$ is at most $k \cdot w$, and all of our previous operations take time proportional to the depth of the tree, we can conclude that the running time of the $\dsunion$ operation is in $\cO(\log(k\cdot w))$ for each call.

    As we stated before, although the method we just discussed works and has a running time in $\cO(\log(k\cdot w))$ for each call, it is not a fully persistent implementation, since we are removing nodes from the leaves when they are not producing an output and we are also modifying the direction bits of the nodes. To solve this problem, we can use the \emph{path copying method}. With this method, whenever we need to modify a node, we create a copy with the modifications applied instead.
    
    In our case, for every insertion we will create a copy of the entire path from the root to the new node, since we will modify the direction bit of each of these nodes, setting the modified copy of the root as the new root of the data structure. It is easy to see that with clever use of pointers, the copying of a node can be done in constant time, so the usage of this method does not increase the overall running time of the $\dsunion$ operation. 
\end{proof}

\subsection*{Proof of Proposition~\ref{prop:algo-correctness}}
\begin{proof}
    Fix a time window size $\window \in \nat$ a stream $\Stream$, a position~$\ipos \in \bbN$ and an \acrocea with equality predicates $\ACH = (Q, \uncq, \bincq, \Omega, \Delta, F)$. The output of the automaton $\ACH$ over $\Stream$ at position $\ipos$ with time window $\window$ is defined as the set of valuations:
    \[
    \sem{\ACH}_\ipos^w(\Stream) \ = \  \{\nu_\rho \mid \text{$\rho$ is an accepting run of $\ACH$ over $S$ at position $\ipos$} \wedge |\ipos - \min(\nu)| \leq \window\}
    \]

    We need to prove that Algorithm~\ref{alg:parallel-eval} enumerates every valuation $\nu_\rho$ without repetitions. One way to do this is showing that at any position in the stream the indices in $H$ contain the information of every single run of $\ACH$ so far showing that the outputs for each of these runs can be enumerated. 
    
    \begin{itemize}
        
        \item Let $\ipos = 0$ and suppose that $\Stream = \bleft R(\bar{x}), \ldots \bright$. $H$ trivially contains the information of all the runs up to this point, so we need to show that this condition still holds after the first tuple.
        
        Looking at the algorithm, after the $\textsc{Reset}$ call, we start with $i=0$, $\DSw = \emptyset$, $N_p = \emptyset$ for every $p\in Q$. Calling $\textsc{FireTransitions}(R(\bar{x}), 0)$ we check each transitions satisfied by $R(\bar{x})$ and we register them in nodes for $\DSw$. Since$\ACH$ is unambiguous, there is only one transition that can lead to an accepting state, $e_f = (\emptyset, U, \emptyset, L_f, p_f) \in \Delta$ with $p_f \in F$ and for $e_f$ we have $N = \{\}$ and $N_{p_f} = \dsextend(L_f, 0, \{\})$. In addition, $\ACH$ can take (several)transitions that do not lead to a final state; these would be transitions of the form $e = (\emptyset, U, \emptyset, L, p) \in \Delta$ with $N_{p} = \dsextend(L, 0, \{\})$.

        On the other hand, $\textsc{UpdateIndices}(R(\bar{x}))$ uses the nodes created in $\textsc{FireTransitions}$ and assigns them to every possible transition that could be satisfied by them. In particular, for every reached state $p$, we add each node in $\nsetq{p}$ to the data structure $H[e, p, \leftbinfunc{\binfunc}_p(t)]$, registering every incomplete run of $\ACH$.
        
        Finally, we enumerate the outputs of each run that reached a final state. Since $\ACH$ is unambiguous, this enumeration will not have duplicates. It is easy to see that enumerate will output our only valuation since $p_f \in F$.
        
        \item Suppose that $H$ contains the information of every single run of $\ACH$ up until position $\ipos -1$ and that $\Stream[i] = S(\bar{y})$ and that we can enumerate every valuation in $\sem{\ACH}_{\ipos-1}^w(\Stream)$. We want to prove that after calling $\textsc{FireTransitions}(S(\bar{y}), \ipos)$ and $\textsc{UpdateIndices}(S(\bar{y}))$, $H$ will also contain the information of the runs up until $\ipos$.
        
        Once again we start with $\nsetq{p} = \emptyset$ for each $p \in Q$, but this time $H[e, p, \leftbinfunc{\binfunc}_p(t)]$ is not empty. Similar to the previous case, upon calling $\textsc{FireTransitions}(S(\bar{y}), \ipos)$, we create a new node for every new state reached by any of the runs and it is clear by the definition of $\Delta$ and $\leftbinfunc{\binfunc}_p(t)$ that $S(\bar{y}) \in U$ and $\bigwedge_{p \in P} \htable{e}{p}{\rightbinfunc{\binfunc}_p(t)} \neq \emptyset$ for a transition $e = (P, U, \binfunc, L, q) \in \Delta$ iff there is a run tree $\rho$ and a node $\bar{u}$ such that $\rho(\bar{u}) = (q, \ipos, L)$.

        In the same fashion, $\textsc{UpdateIndices}(S(\bar{y}))$ will thoroughly calculate for each transition and each state in those transitions the left projection of the binary predicate for $S(\bar{y})$, maintaining the data structure in $H$ updated with the runs od $\ACH$.

        Finally, the algorithm was already capable of enumerating every valuation that ends in a position $j < i$, and we get from $\textsc{FireTransitions}(S(\bar{y}), \ipos)$ that every new accepting run will have its associated nodes.      
    \end{itemize} \vspace{-5mm}
\end{proof} 	

\end{document}